\documentclass[runningheads, envcountsame, a4paper]{llncs}

\usepackage{amsmath}
\usepackage{amssymb}

\usepackage{microtype}
\usepackage{graphicx}
\usepackage{caption}

\newcommand{\lvec}[1]{\overset{{}_{\leftarrow}}{#1}}
\newcommand{\lrange}[1]{\overleftarrow{#1}}
\newcommand{\lcp}{\mathsf{lcp}}
\newcommand{\rlcp}{\lrange{\mathsf{lcp}}}
\newcommand{\suff}{\mathrm{suf}}
\newcommand{\pref}{\mathrm{pre}}

\begin{document}

\title{Detecting One-variable Patterns}
\author{Dmitry Kosolobov\inst{2} \and Florin Manea\inst{1} \and Dirk Nowotka\inst{1}}
\institute{Institut f\"{u}r Informatik, Christian-Albrechts-Universit\"{a}t zu Kiel, Kiel, Germany\\
\email{$\{$flm,dn$\}$@informatik.uni-kiel.de}\and
Institute of Mathematics and Computer Science, Ural Federal University, Ekaterinburg, Russia\\
\email{dkosolobov@mail.ru}}

\titlerunning{Detecting One-variable Patterns}
\authorrunning{D. Kosolobov, F. Manea, and D. Nowotka}
\toctitle{Detecting One-variable Patterns}
\tocauthor{Dmitry~Kosolobov, Florin~Manea, and Dirk~Nowotka}

\maketitle
\setcounter{footnote}{0}

\begin{abstract}
Given a pattern $p = s_1x_1s_2x_2\cdots s_{r-1}x_{r-1}s_r$ such that $x_1,x_2,\ldots,x_{r-1}\in\{x,\lvec{x}\}$, where $x$ is a variable and $\lvec{x}$ its reversal, and $s_1,s_2,\ldots,s_r$ are strings that contain no variables, we describe an algorithm that constructs  in $O(rn)$ time a compact representation of all $P$ instances of $p$ in an input string of length $n$ over a polynomially bounded integer alphabet, so that one can report those instances in $O(P)$ time.
\end{abstract}

\section{Introduction}
A \emph{pattern} is a string consisting of \emph{variables} (e.g., $x,y,z$) and \emph{terminal letters} (e.\,g., $a,b,c$). The terminal letters are treated as constants, while the variables are letters to be uniformly replaced by strings over the set of terminals (i.\,e., all occurrences of the same variable are replaced by the same string); by such a replacement, a pattern is mapped to a terminal string. Patterns with variables appeared in various areas of computer science, e.g., stringology and pattern matching~\cite{ami:gen}, combi\-natorics on words~\cite{lot:alg:unPat}, language and learning theory~\cite{ang:fin2}, or regular expressions with back references~\cite{fri:mas,Schmid16}, used in programming languages like Perl, Java, Python. In such applications, patterns are used to express string searching questions such as testing whether a string contains regularities.

Here, we consider the so-called \emph{one-variable patterns}~$p = s_1x_1\cdots s_{r-1}x_{r-1}s_r$ such that, for all $z$, $x_z\in \{x,\lvec{x}\}$, where $x$ is a variable and $\lvec{x}$ its reversal, and $s_z$ is a string over a set $\Sigma$ of terminals. An \emph{instance} of $p$ in a text $t$ is a substring $ s_1w_1\cdots s_{r-1}w_rs_r$ of $t$,~with~$w_z = w$ if $x_z = x$ and $w_z=\lvec{w}$ if  $x_z = \lvec{x}$, for a non-empty $w\in \Sigma^*$ called \emph{substitution} of $x$. We address the problem of efficiently finding instances of such patterns in~texts.

For example, let $p=axabxbc\lvec{x}$. An instance of this pattern, if the alphabet of terminals is $\{a,b,c\}$, is $a\,abc\,ab\,abc\,bc\,cba$, where $x$ is substituted by $abc$ (and, consequently, $\lvec{x}$ by $cba$). Another instance is $a\,aaabbb\,ab\,aaabbb\,bc\,{bbbaaa}$ if $x$ is substituted by $aaabbb$. Both these instances occur in the text $t= aabcababcbccb aaaabbbabaaabbbbcbbbaaa$: the former instance starts at position $1$ and the later starts at position $14$. These two instances overlap at position $14$.

Our motivation for studying such patterns is two-fold. Firstly, the efficient matching of several classes of restricted patterns was analyzed in \cite{FernauManeaMercasSchmid} and connected to algorithmic learning theory \cite{FeMaMeSc14-TCS}. Generally, matching patterns with variables to strings is NP-complete \cite{ehr:fin}, so it seemed an interesting problem to find structurally restricted classes of patterns for which the matching problem is tractable. As such, finding all occurrences of a one-variable pattern in a word occurred as basic component in the matching algorithms proposed in~\cite{FernauManeaMercasSchmid} for patterns with a constant number of repeated variables or for non-cross patterns (patterns that do not have the form $..x..y..x..$).

Secondly, our work extends the study of pseudo-repetitions (patterns from $\{x,\lvec{x}\}^*$). The concept of pseudo-repetitions (introduced in \cite{CCKS09}, studied from both combinatorial \cite{MMNS14} and algorithmic \cite{Xu10,GawrychowskiManeaNowotka} points of view) draws its original motivations from important biological concepts: tandem repeat, i.e., a consecutive repetition of the same sequence of nucleotides; inverted repeat, i.e., a sequence of nucleotides whose reversed image occurred already in the longer DNA sequence we analyze, both occurrences (original and reversed one) encoding, essentially, the same genetic information; or, hairpin structures in the DNA sequences, which can be modeled by patterns of the form $s_1 xs_2 \lvec{x} s_3$. More interesting to us, from a mathematical point of view, pseudo-repetitions generalize both the notions of repetition and of palindrome, central to combinatorics on words and applications. The one-variable pattern model we analyze generalizes naturally the mathematical model of pseudo-repetition by allowing the repeated occurrences of the variable to be separated by some constant factors.

Thus, we consider the next problem, aiming to improve the detection of pseudo-repetition \cite{GawrychowskiManeaNowotka}, as well as a step towards faster detection of occurrences of restricted~patterns~\cite{FernauManeaMercasSchmid,FeMaMeSc14-TCS}.
\begin{problem}\label{main_pb}
Given a string $t\in \Sigma^*$ of length $n$ and a pattern  $p = s_1x_1\cdots s_{r-1}x_{r-1}s_r$ such that, for $1\leq z \leq r-1$, $x_z\in \{x,\lvec{x}\}$ where $x\notin\Sigma$ is a variable and $\lvec{x}$ its reversal, and $s_z\in \Sigma^*$ for $1\leq z \leq r$, report all $P$ instances of $p$ in $t$ (in a form allowing their retrieval in $O(P)$~time).
\end{problem}
We assume that $t$ and all strings $s_z$, for $z=1,\ldots, r$, are over an integer alphabet $\Sigma = \{0,1,\ldots,n^{O(1)}\}$, and that we use the word RAM model with $\Theta(\log n)$-bit machine words\footnote{Hereafter, $\log$ denotes the logarithm with base $2$.} (w.l.o.g., assume that $\log n$ is an integer). In this setting, we propose an algorithm that reports in $O(rn)$ time all instances of $p$ in $t$ in a compactly encoded form, which indeed allows us to retrieve them in $O(P)$ time. Our approach is based on a series of deep combinatorics on words observations, e.g., regarding the repetitive structure of the text, and on the usage of efficient string-processing data structures, combining and extending in novel and non-trivial ways the ideas from~\cite{FernauManeaMercasSchmid,GawrychowskiManeaNowotka,KosolobovRepetitions}.

If the pattern contains only a constant number of variables (e.g., generalized squares or cubes with terminals between the variables), our algorithm is asymptotically as efficient as the algorithms detecting fixed exponent (pseudo-)repetitions. For arbitrary patterns, our solution generalizes and improves the results of \cite{GawrychowskiManeaNowotka}, where an $O(r^2 n)$-time solution to the problem of finding one occurrence of a one-variable pattern with reversals (without terminals) was given. Here, compared to \cite{GawrychowskiManeaNowotka}, we work with patterns that contain both variables and terminals and we detect, even faster, all their instances. Also, we improve the results of \cite{FernauManeaMercasSchmid} in several directions: as said, we find all instances of a one-variable pattern (in~\cite{FernauManeaMercasSchmid} such a problem was solved as a subroutine in the algorithm detecting non-cross patterns, and only some instances of the patterns were found), our algorithm is faster by a $\log n$ factor, and our patterns also contain reversed variables.

In this paper, we omit most of the technicalities of the solution to Problem~\ref{main_pb} from the main part, and prefer to keep the presentation at an intuitive level; the full proofs are available in Appendix. 

\section{Preliminaries}

Let $w$ be a string of length $n$. Denote $|w| = n$. The \emph{empty string} is denoted by $\epsilon$. We write $w[i]$ for the $i$th letter of $w$ and $w[i..j]$ for $w[i]w[i{+}1]\cdots w[j]$. A string $u$ is a \emph{substring} of $w$ if $u = w[i..j]$ for some $i\leq j$. The pair $(i,j)$ is not necessarily unique; we say that $i$ specifies an \emph{occurrence} of $u$ in $w$. A substring $w[1..j]$ (resp., $w[i..n]$) is a \emph{prefix} (resp. \emph{suffix}) of $w$. The \emph{reversal of $w$} is the string $\lvec{w} = w[n]\cdots w[2]w[1]$; $w$ is a \emph{palindrome} if $w = \lvec{w}$. For any $i,j \in \mathbb{R}$, denote $[i..j] = \{k\in \mathbb{Z} \colon i \le k \le j\}$, $(i..j] = [i..j] \setminus \{i\}$, $[i..j) = [i..j] \setminus \{j\}$, $(i..j) = [i..j)\cap (i..j]$. Our notation for arrays is similar to that for strings, e.g., $a[i..j]$ denotes an array indexed by the numbers $[i..j]$: $a[i], a[i{+}1],\ldots, a[j]$.

In Problem \ref{main_pb} we are given an input string (called text) $t$ of length $n$ and a pattern $p = s_1x_1s_2x_2\cdots s_{r-1}x_{r-1}s_r$ such that, for $z \in [1..r)$, $x_z \in \{x,\lvec{x}\}$ and $s_1, s_2, \ldots, s_r$ are strings that contain no $x$ nor $\lvec{x}$. For the simplicity of exposure, we can assume $x_1=x$. An \emph{instance of $p$} in the text $t$ is a substring $t[i..j] = s_1w_1s_2w_2\cdots s_{r-1}w_{r-1}s_r$ such that, for $z \in [1..r)$, $w_z = w$ if $x_z = x$, and $w_z = \lvec{w}$ if $x_z = \lvec{x}$, where $w$ is a string called a \emph{substitution of $x$}; $\lvec{w}$ is called a \emph{substitution of $\lvec{x}$}. We want to find all instances of $p$ occurring in $t$.

An integer $d > 0$ is a \emph{period} of a string $w$ if $w[i] = w[i{+}d]$ for all $i \in [1..|w|{-}d]$; $w$ is \emph{periodic} if it has a period ${\le}\frac{|w|}{2}$. For a string $w$, denote by $\pref_d(w)$ and $\suff_d(w)$, respectively, the longest prefix and suffix of $w$ with period $d$. A \emph{run} of a string $w$ is a periodic substring $w[i..j]$ such that both substrings $w[i{-}1..j]$ and $w[i..j{+}1]$, if defined, have strictly greater minimal periods than $w[i..j]$. A string $w$ is \emph{primitive} if $w \ne v^k$ for any string $v$ and any integer $k > 1$.

\begin{lemma}[see~\cite{CrochemoreRytter}]
A primitive string $v$ occurs exactly twice in the string $vv$.\label{PrimitiveVV}
\end{lemma}

\begin{lemma}\label{sum_runs}
Let $R$ be the set of all runs of $t$, whose period is at least three times smaller than the length of the run (such runs are called cubic). Then $\sum_{s\in R} |s|\in O(n \log n)$.
\end{lemma}
\begin{proof}
Consider a run $t[i..j] \in R$ with the minimal period $p$. Since a primitively rooted square of length $2p$ occurs at any position $k \in [i..j{-}2p{+}1]$, the sum $\sum_{s\in R} |s|$ is upper bounded by three times the number of primitively rooted squares occurring in $t$. At each position of $t$, at most $2 \lceil\log n\rceil$ primitively rooted squares may occur (see, e.g., \cite{CrochemoreRytter}), so the result follows.\qed
\end{proof}

In solving Problem \ref{main_pb}, we use a series of preprocessing steps. First, we find all runs in $t$ in $O(n)$ time using the algorithm of~\cite{BannaiIInenagaNakashimaTakedaTsuruta} and, using radix sort, construct lists $R_d$, for $d = 1,2,\ldots,n$, such that $R_d$ contains the starting positions of all runs with the minimal period $d$ in increasing order. We produce from $R_d$ two sublists $R'_d$ and $R''_d$ containing only the runs with the lengths ${\ge}\log n$ and $\geq \log\log n$, respectively (so that $R'_d$ is a sublist of $R''_d$). The following lemma provides us fast access to the lists $R_d, R'_d, R''_d$ from periodic substrings of $t$.

\begin{lemma}[{\cite[Lemma 6.6]{KociumakaRadoszewskiRytterWalen}}]
With $O(n)$ time preprocessing, we can decide in $O(1)$ time for any substring $t[i..j]$ of $t$ whether it is periodic and, if so, compute its minimal period $d$ and find in $R_d$, $R'_d$, or $R''_d$ the run containing $t[i..j]$.\label{SubstringRun}
\end{lemma}

For $i, j \in [1..n]$, let $\lcp(i,j)$ and $\rlcp(i,j)$ be the lengths of the longest common prefixes of the strings $t[i..n]$, $t[j..n]$ and $\lrange{t[1..i]}$, $\lrange{t[1..j]}$, respectively. In $O(n)$ time we build for the string $t\lvec{t}$ the \emph{longest common prefix} data structure (for short, called the \emph{$\lcp$ structure}) that allows us retrieving the values $\lcp(i, j)$ and $\rlcp(i,j)$ for any $i, j \in [1..n]$ in $O(1)$ time (see~\cite{KaSaBu06,CrochemoreRytter}). Thus, to check if the substrings of length $\ell$ starting (resp., ending) at positions $i$ and $j$ in the string $t$ are equal, we just check whether $\lcp(i,j)\geq \ell$ (resp., ${\rlcp}(i,j)\geq \ell$). As a side note, we essentially use that we can compare the reversed image of two substrings of $t$ using the $\lcp$ structure built for $t\lvec{t}$.

With the $\lcp$ structure, it is easy to solve Problem~\ref{main_pb} in $O(rn^2)$ time: we first apply any linear pattern matching algorithm to find in $O(rn)$ time all occurrences of the strings $s_1, s_2, \ldots, s_r$ in $t$ and then, for every position $i \in [1..n]$ of $t$ and every $\ell \in [0..n]$, we check in $O(r)$ time whether an instance $s_1w_1\cdots s_{r-1}w_{r-1}s_r$ of the pattern $p$, with $\ell = |w_1| = \cdots = |w_{r-1}|$, occurs at position $i$.

\paragraph{General strategy.}
For each $z \in [1..r]$, using a pattern matching algorithm (see~\cite{CrochemoreRytter}), we fill in $O(n)$ time a bit array $D_z[1..n]$ where, for $i \in [1..n]$, $D_z[i] = 1$ iff $s_z$ occurs at position $i$. Assume that $p$ contains at least two occurrences of the variable, i.e., $p\notin\{ s_1xs_2\}$ (in the case $p = s_1xs_2$ each instance of $p$ is given by an occurrence of $s_1$, stored in $D_1$, followed by an occurrence of $s_2$, stored in $D_2$).

Let $\alpha = \frac{4}{3}$. For each $k \in [0..\log_{\alpha} n]$, our algorithm finds all instances of $p$ that are obtained by the substitution of $x$ with strings of lengths from $(\frac{3}2\alpha^k..2\alpha^k]$. Clearly, the intervals $(\frac{3}2\alpha^k..2\alpha^k]$ do not intersect and their union covers the interval $[2..n]$. In this manner, our algorithm obtains all instances of $p$ with substitutions of $x$ of length at least two. The remaining instances, when the string substituting $x$ has length one or zero, can be easily found in $O(rn)$ time using the arrays $\{D_z\}_{z=1}^r$.

So, let us fix $k \in [0..\log_{\alpha} n]$ and explain our strategy for this case. Suppose that, for $i,j \in [1..n]$, $t[i..j] = s_1w_1s_2w_2\cdots s_{r-1}w_{r-1}s_r$ is an instance of $p$ and $\frac{3}2\alpha^k < |w_1| = \cdots = |w_{r-1}| \le 2\alpha^k$; then $w_1$ contains a substring $v$ of length $\lceil\alpha^k\rceil$ starting, within $t$, either at position $q_1=h\lceil\alpha^k\rceil + 1$ or at position $q_1=h\lceil\alpha^k\rceil + \lfloor\frac{\lceil\alpha^k\rceil}2\rfloor$ for some integer $h \ge 0$. Based on this observation, we consider all choices of a  substring $v$ of $t$, with length $\lceil\alpha^k\rceil$, starting at positions $h\lceil\alpha^k\rceil + 1$ and $h\lceil\alpha^k\rceil + \lfloor\frac{\lceil\alpha^k\rceil}2\rfloor$ for $h \ge 0$. Such a string $v$ acts as a sort of anchor: it restricts (in a strong way, because of its rather large length with respect to $|w_1|$) the positions where  $w_1$ may occur in $t$, and copies of either $v$ or $\lvec{v}$ should also occur in all $w_2,\ldots,w_{r-1}$, thus restricting the positions where these strings may occur in $t$, as well. Based on a series of combinatorial observations regarding the way such substrings $v$ occur in $t$, and using efficient data structures to store and manipulate these occurrences, we find all corresponding instances of $p$ that contain $v$ in the substitution of $x_1$ in $O(r + \frac{r|v|}{\log n})$ time plus $O(\frac{\log n}{\log\log n})$ time if $\frac{\log n}{16\log\log n} \le |v| \le \log n$. We discuss two cases: $v$ is non-periodic or~periodic.

In the first case, distinct occurrences of $v$ (or $\lvec{v}$) in $t$ do not have large overlaps, so we can detect them rather fast, as described in  Lemma \ref{vFind}: for $\lambda = |s_2|$, we preprocess a data structure that allows us to efficiently find all occurrences of $v$ or $\lvec{v}$ at the distance $\lambda$ to the right of $v$ and these occurrences serve as additional anchors inside the substitution $w_2$; note that the case of very short $v$ requires a separate discussion. Hence, the distinct instances of $p$ where the substitution of $x$ contains a certain non-periodic $v$ also do not have large overlaps (which means, as well, that they are not too many), and they can be identified (and stored, as described in Lemma \ref{FixedLen}) by trying to align occurrences of the strings $s_1,\ldots,s_r$ in a correct manner around the found $v$'s.

Then we consider the case when $v$ is periodic. Then, the occurrences of $v$ or $\lvec{v}$ corresponding to different instances of $p$ might have large overlaps and form runs, so we analyze the runs structure of $t$. Consider, for the simplicity of exposure, a typical example: $t = (abc)^m$ contains $\Theta(|t|^2)$ instances of $p = xcxcabcxcxcxca$ with substitutions $x = ab(cab)^k$, for different $k$. The point in this example is that almost all substitutions are periodic and are contained in one run with the same minimal period. We can encode these instances by an arithmetic progression: for all $ 0\leq h\leq m{-}7$, $0\leq k \leq m{-}h{-}7$, there is an instance of $p$ starting at position $1{+}3h$ of $t$ with substitution of length $2{+}3k$. It turns out, as described in Lemmas~\ref{AllInRunOccurrences} and~\ref{GeneralInRunLemma}, that, for any pattern $p$, all instances of $p$ whose substitutions are periodic substrings of one run with the same minimal period can be encoded by similar arithmetic~progressions.

Consider now another relevant example: $t = (abc)^{\ell}d(abc)^m$ contains $\Theta(|t|)$ instances of $p = xxdxabcxx$ with substitutions $x = (abc)^k$. All these instances can be encoded as follows: for all $k = 0,1,\ldots,\min\{\ell,m\}$, there is an instance of $p$ starting at position $1{+}3\ell{-}3k$ with substitution of length $3k$. So, the letter $d$ ``separates'' the image of $p$ into two runs, breaking the period of the first run. As shown in Lemmas~\ref{SepCond} and~\ref{subz}, there might exist only a constant number of such ``separators'' in a general $p$ and all instances of $p$, with the image $x$ periodic, and which lie in two runs with the same minimal period, split by a given ``separator'', can be encoded by similar arithmetic progressions (the analysis of this case is similar to the analysis of in-a-run instances, so, it is moved in Appendix).

If the substitutions in an instance of $p$ lie in three or more runs (so, also there are more points where the period breaks inside each instance of $p$), then we can find the possible occurrences of $v$ (which are periodic, so they must avoid period-breaking points that separate the runs contained in $p$) and, consequently, find the instances of $p$. The combinatorics of such instances of the pattern is discussed in Lemmas \ref{wsw}, \ref{wsw2}: the essential idea is that the occurrences of $v$ and $\lvec{v}$ in $p$ and the substrings connecting them form runs, separated by substrings which break the periodicity; these substrings should correspond to substrings that interrupt runs in $t$. The actual algorithm identifying and storing these instances of the pattern follows from Lemmas~\ref{SeparateLemma},~\ref{SepCond},~\ref{subz} (and the comments connecting~them).

Finally, since there are $O(\frac{n}{\alpha^k})$ such substrings $v$ and at most $O(n / \frac{\log n}{\log\log n})$ of them (for all $k = 0,1,\ldots$ in total) are such that $\frac{\log n}{16\log\log n} \le |v| \le \log n$, the overall time is $O(\sum_{k=0}^{\log_{\alpha} n} \frac{n}{\alpha^k} (r + \frac{r\alpha^k}{\log n}) + (n / \frac{\log n}{\log\log n})\frac{\log n}{\log\log n}) = O(rn)$.

The details of all the cases considered in our approach are given in Sections \ref{SectNonPeriodic} and \ref{SectPeriodic}, following the general strategy described above. Summing up, we~get:
\begin{theorem}
Problem \ref{main_pb} can be solved in $O(rn)$ time.\label{MainResult}
\end{theorem}

\section{Non-periodic Anchor Substring $v$}\label{SectNonPeriodic}

As described in the \emph{General Strategy} paragraph, we first choose an anchor string~$v$ occurring in $w_1$ and then try to construct an instance of the pattern $p$ around this $v$. So, let $v$ be a substring of $t$ of length $\lceil\alpha^k\rceil$ starting at position $q_1 = h|v| + 1$ for some integer $h \ge 0$ (the case of position $h|v| + \lfloor\frac{|v|}2\rfloor$ is similar). As explained before, we will iterate through all possible values of $h$, which allows us to identify all instances of the pattern. For a fixed $v$, using Lemma~\ref{SubstringRun}, we check whether it is periodic. In this section, we suppose that $v$ is not periodic; the case of periodic $v$ is considered in Section~\ref{SectPeriodic}.

Our aim is to find all instances $t[i..j] = s_1w_1s_2w_2\cdots s_{r-1}w_{r-1}s_r$ of $p$ in which $w_1$ contains $v$ and has length close to $|v|$, i.e., $i$ and $j$ must be such that $i + |s_1| \le q_1 < q_1 + |v| \le i + |s_1w_1|$ and $\frac{3}2|v| < |w_1| = \cdots = |w_{r-1}| \le 2|v|$.

Let $t[i..j]$ be such a substring. It follows from the inequality $\frac{3}2|v| < |w_1| \le 2|v|$ that we can compute a relatively small interval of $t$ where the $v$ (or $\lvec{v}$) corresponding to $w_2$ may occur. More precisely, if $w_1 = w_2$ (resp., $w_1 = \lvec{w_2}$), then the string $v$ (resp.,~$\lvec{v}$) has an occurrence starting at a position from the interval $[q_1 + |vs_2| .. q_1 + |vvs_2v|]$. Since $v$ is not periodic, the length of the overlap between any two distinct occurrences of $v$ is less than $\frac{|v|}2$. Hence, there are at most four occurrences of $v$ (resp., $\lvec{v}$) starting in $[q_1 + |vs_2| .. q_1 + |vvs_2v|]$. To find these occurrences, our algorithm applies the following general lemma for $\lambda = |s_2|$.

\begin{lemma}
Let $\lambda \ge 0$ be an integer. We can preprocess the text $t$ of length $n$ in $O(n)$ time to produce data structures allowing us to retrieve, for any given non-periodic substring $v = t[q..q'{-}1]$, all occurrences of $v$ and $\lvec{v}$ starting in the substring $t[q' + \lambda .. q' + \lambda + 2|v|]$ in:
\begin{itemize}
\item  $O(\frac{|v|}{\log n})$ time if $|v| > \log n$,
\item $O(\frac{\log n}{\log\log n})$ time if $\frac{\log n}{16\log\log n} \le |v| \le \log n$, and
\item$O(1)$ time otherwise.
\end{itemize}\label{vFind}
\end{lemma}
\begin{proof}
For $i \in [1..n]$, let $t_i = t[i..i{+}\log n-1]$ be the substring of length $\log n$ starting at position $i$ in $t$. Let $S$ be the set of all distinct strings $t_i$.
Using the suffix array of~$t$, its $\lcp$ structure, and radix sort, we construct in $O(n)$ time the set of arrays $\{A_s\}_{s \in S}$ such that, for any $s\in S$, $A_s$ contains the starting positions of all occurrences of $s$ in $t$ in ascending order. Essentially, for each $s\in S$, we locate an occurrence of $s$ in $t$ and then produce a ``cluster'' of the suffix array of $t$ with the suffixes starting with  $s$, then we radix sort (simultaneously) the positions in these ``clusters'' (all numbers between $1$ and $n$, keeping track of the ``cluster'' from where each position came), to obtain the arrays $A_s$. Further, using the suffix array of the string $t\lvec{t}$, its $\lcp$ structure, and radix sort, we build in $O(n)$ time arrays of pointers $B[1..n]$ and $\lvec{B}[1..n]$ such that, for $i \in [1..n]$, $B[i]$ (resp., $\lvec{B}[i]$) points to the element of $A_{t_i}$ (resp., $A_{\lvec{t_i}}$) storing the leftmost position $j$ with $j \ge i + \lambda$ and $t_i = t_j$ (resp., $\lvec{t_i} = t_j$);
$B[i]$ (resp., $\lvec{B}[i]$) is undefined if there is no such $j$.

\emph{The case $|v| > \log n$.}
In this case, to find all required occurrences of $v$, we note that $v=t[q..q'-1]$ starts with $t_q$. Thus, we first find all occurrences of $t_q$ starting within the segment $[q + \lambda .. q' + \lambda + 2|v|]$. The sequence of all such occurrences forms a contiguous subarray in $A_{t_q}$ and $B[q]$ points to the beginning of this subarray.

In a first case, suppose that the distance between any two consecutive positions stored in this subarray is greater than $\frac{|t_q|}2$. Then there are at most $O(\frac{|v|}{|t_q|}) = O(\frac{|v|}{\log n})$ such occurrences of $t_q$. Some of these occurrences may be extended to form an occurrence of $v$, and they must be identified. To check in constant time whether $v$ occurs indeed at a given position $\ell$ of the subarray we use the $\lcp$ structure and verify whether $\lcp(\ell,q)\geq |v|$.

The case of the string $\lvec{v}$ is analogous but involves $\lvec{t_q}$ and $\lvec{B}$ instead of $t_q$ and~$B$. Hence, we find all required occurrences of $v$ and $\lvec{v}$ in $O(\frac{|v|}{\log n})$ time.

Suppose that the aforementioned subarray of $A_{t_q}$ (resp., $A_{\lvec{t_q}}$), containing the positions of $t_q$ (resp., $\lvec{t_q}$) in the desired range, contains two consecutive occurrences of $t_q$ (resp.,~$\lvec{t_q}$) whose starting positions differ by at most $\frac{|t_q|}2$. Then $t_q$ is periodic. Using Lemma~\ref{SubstringRun}, we compute the minimal period $d$ of $t_q$ and find, in $O(1)$ time, the run $t[i'..j']$ (in the list $R'_d$) containing $t_q$. Recall now that $v$ is not periodic, so we must have that $t[q..j']$ is $\pref_d(v)$, the maximal $d$-periodic prefix of $v$, and $|\pref_d(v)| < |v|$. We now focus on finding the occurrences of $t_q$ in the range $[q' + \lambda .. q' + \lambda + 2|v|]$. Since $R'_d$ contains only runs of length ${\ge}\log n$ and any two runs with period $d$ cannot overlap on more than $d-1$ letters, there are at most $O(\frac{|v|}{\log n})$ runs in $R'_d$ that overlap with the segment $[q + \lambda .. q' + \lambda + 2|v|]$. These runs can be all found in $O(\frac{|v|}{\log n})$ time. Some of them may end with $\pref_d(v)$ and may be extended to the right to obtain an occurrence of $v$ (resp.,~$\lvec{v}$). If $t[i''..j'']$ is one of the runs we found, then there might be an occurrence of $v$ starting at position $j'' - j' + q$ or an occurrence of $\lvec{v}$ ending at position $i'' + j' - q$. So, using the $\lcp$ structure, in a similar way as before, we find all required occurrence of $v$ (resp.,~$\lvec{v}$) in $O(\frac{|v|}{\log n})$ time.

It remains to consider how to find all occurrences of $v = t[q..q'{-}1]$ (resp., $\lvec{v}$) starting in the segment $[q'{+}\lambda .. q'{+}\lambda{+}2|v|]$ in the case $\frac{\log n}{16\log\log n} \le |v| \le \log n$ and $|v| < \frac{\log n}{16\log\log n}$.

\emph{The case $\frac{\log n}{16\log\log n} \le |v| \le \log n$.}
This case is similar to the case $|v| > \log n$. For $i \in [1..n{-}\lfloor\log\log n\rfloor]$, define $t'_i = t[i..i{+}\lfloor\log\log n\rfloor]$. Let $S'$ be the set of all distinct strings $t'_i$. In the same way as in the case $|v| > \log n$, using the suffix array of~$t$, its $\lcp$ structure, and radix sort, we construct in $O(n)$ time the set of arrays $\{A'_{s'}\}_{s'\in S'}$ such that, for any $s'\in S'$, $A'_{s'}$ contains the starting positions of all occurrences of $s'$ in $t$ in ascending order. Further, using the suffix array of the string $t\lvec{t}$, its $\lcp$ structure, and radix sort, we build in $O(n)$ time arrays of pointers $B'[1..n]$ and $\lvec{B'}[1..n]$ such that, for $i \in [1..n]$, $B'[i]$ (resp., $\lvec{B'}[i]$) points to the element of $A'_{t'_i}$ (resp., $A'_{\lvec{t'_i}}$) storing the leftmost position $j$ with $j \ge i + \lambda$ and $t'_i = t'_j$ (resp., $\lvec{t'_i} = t'_j$); $B'[i]$ (resp., $\lvec{B'}[i]$]) is undefined if there is no such $j$. Now we proceed like in the case $|v| > \log n$ but use $t'_q$ instead of $t_q$, the arrays $A'_{t'_q}$, $B'$, $\lvec{B'}$ instead of $A_{t_q}$, $B$, $\lvec{B}$, and the list $R''_d$ instead of $R'_d$. The processing takes $O(\frac{|v|}{\log\log n}) = O(\frac{\log n}{\log\log n})$ time.

\emph{The case $|v| < \frac{\log n}{16\log\log n}$.}
Using radix sort, we can reduce the alphabet of $t$ to $[0..n)$ in $O(n)$ time; let $\$$ be a new letter. For $h \in [0..\frac{n}{\log n})$, let $e_h = t[h\log n{+}1..h\log n{+}2\log n]$ and $f_h = t[h\log n{+}\lambda..h\log n{+}\lambda{+}5\log n]$ assuming $\$ = t[n{+}1] = t[n{+}2] = \ldots$, so that $e_h$ and $f_h$ are well defined. Note that $v$ is a substring of $e_h$ for $h = \lfloor\frac{q - 1}{\log n}\rfloor$ and, if there is an occurrence of $v$ (resp., $\lvec{v}$) starting in the segment $[q' + \lambda .. q' + \lambda + 2|v|]$, then this occurrence is a substring~of~$f_h$.

For each $h \in [0..\frac{n}{\log n})$, our algorithm constructs a string $g_h = e_h\$f_h$ and reduces the alphabet of $g_h$ to $[1..|g_h|]$ as follows. Let $E[0..n]$ be an array of integers filled with zeros. While processing $g_h$, we maintain a counter $c$; initially, $c = 0$. For $i = 1,2,\ldots,|g_h|$, we check whether $E[g_h[i]] = 0$ and, if so, assign $c \gets c + 1$ and $E[g_h[i]] \gets c$. Regardless of the result of this check, we perform $g_h[i] \gets E[g_h[i]]$. Once the alphabet of $g_h$ is reduced, we clear all modified elements of $E$ using an unmodified copy of $g_h$ and move on to $g_{h+1}$. Thus, the reductions of the alphabets of all $g_h$ take $O(n + \sum_{h=0}^{\lfloor n/\log n\rfloor}|g_h|) = O(n)$ overall time.

Each letter in a string $g_h$ fits in $\lceil\log(|g_h|+1)\rceil \le 2\lceil\log\log n\rceil$ bits. Hence, the substrings of $g_h$ corresponding to the substrings $v = t[q..q'{-}1]$ and $t[q' + \lambda .. q' + \lambda + 3|v|]$ together fit in $8|v|\lceil\log\log n\rceil \le \frac{\log n}{2}$ bits. Thus, we can perform the searching of $v$ (resp., $\lvec{v}$) in $t[q' + \lambda .. q' + \lambda + 3|v|]$ in $O(1)$ time using a precomputed table of size $O(2^{\frac{\log n}2}) = O(\sqrt{n})$.
\qed
\end{proof}

Recall that $q_1$ was the starting point of $v$ (for simplicity, assume that $v=t[h_1..h_2]$, where $h_1$ is an alias of $q_1$ that is only used for the uniformity of the notation). Let $q_2 \in [q_1 + |vs_2| .. q_1 + |vs_2vv|]$ be the starting position of an occurrence of $v$ (or $\lvec{v}$) found by Lemma~\ref{vFind}. We now want to see whether there exists an instance of the pattern that has the anchor $v$ from $w_1$ occurring at position $q_1$ and the corresponding $v$ (resp., $\lvec{v}$) from $w_2$ occurring at $q_2$.

If $x_1 = x_2$ (and, consequently, $w_1=w_2$), then $\beta = q_2 - q_1 - |s_2|$ is the length of substitution $w_1$ of $x$ that could produce the occurrence of $v$ at position $q_2$. Once the length $\beta$ is computed, we get that $w_1$ can start somewhere between $h_2-\beta-|s_1|$ and $q_1-|s_1|=h_1-|s_1|$, so all corresponding instances of $p$ will start in the interval $[h_2{-}\beta{-}|s_1|{+}1..h_1{-}|s_1|]$.
These instances (determined by $h_1=q_1$, $|v|$, and $\beta$) can be found by the following lemma (see the case $x_1 \ne x_2$ in Appendix).

For a given $\beta$, let $L_p(\beta)=|s_1s_2\cdots s_r| + (r - 1)\beta $, that is, the length of the image of the pattern $p$ when $x$ is substituted by a variable of length $\beta$.
\begin{lemma}
Given a substring $t[h_1..h_2] = v$ and an integer $\beta \ge |v|$, we can compute a bit array $occ[h_2{-}\beta{-}|s_1|{+}1..h_1{-}|s_1|]$ such that, for any $i$, we have $occ[i] = 1$ iff the string $t[i..i{+}L_p(\beta) - 1]$ is an instance of $p$ containing $v$ in its substring that corresponds to $w_1$ (i.e., $i + |s_1| \le h_1 < h_2 < i + |s_1| + \beta$). This computation takes $O(r + \frac{r\beta}{\log n})$ time, to which we add $O(\frac{\log n}{\log\log n})$ time when $\frac{\log n}{16\log\log n} \le |v| \le \log n$.\label{FixedLen}
\end{lemma}
\begin{proof} The general idea of the proof is as follows. Knowing where $v$ (which anchors $w_1$, which substitutes $x$) starts and knowing the length $|w_1|$, we know, if $x_1 = \cdots = x_{r-1}$, where the corresponding occurrences of $v$ from $w_2,\ldots,w_{r-1}$ should be positioned (the case when $x_i \ne x_j$, for some $i \ne j$, is analyzed using more complicated ideas, e.g., from~\cite{CrochemoreEtAl}; see Appendix). We check, in $O(r)$ time, if they indeed occur at those positions. Suppose this checking succeeds. These $v$'s might correspond to more instances of $p$ as in Fig.~\ref{fig:shades}. We further check where the $w_z$'s corresponding to occurrences of $x$ in $p$ may occur.
\begin{figure}[htb]
\centering
    \includegraphics[scale=0.22]{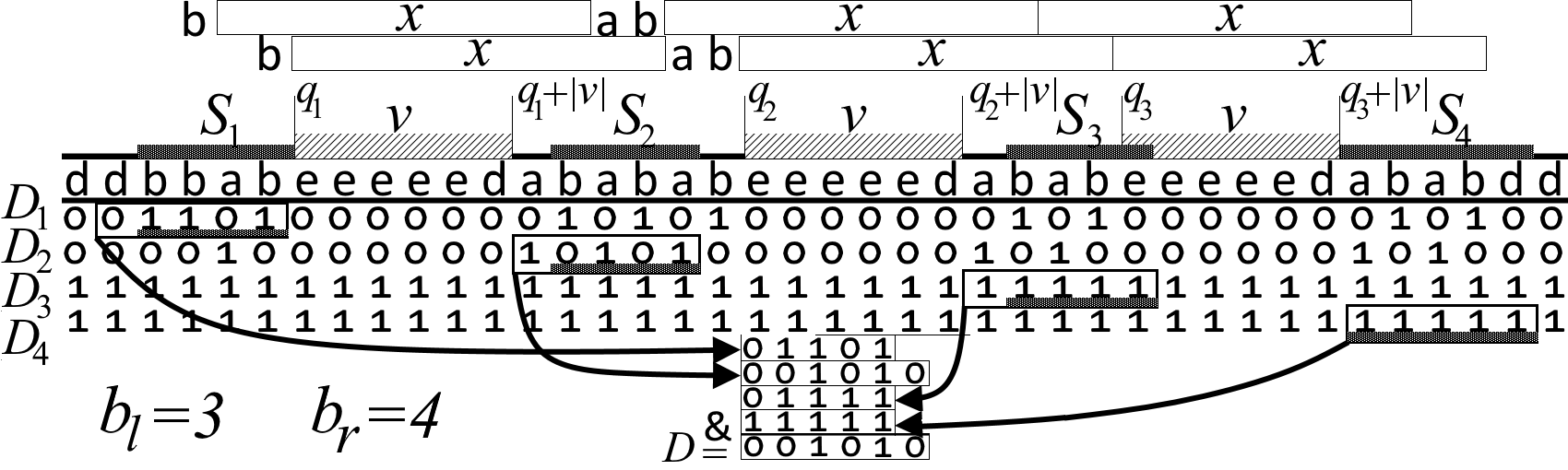}
    \caption{{\small Two instances of the pattern $p = bxabxx$.}}
    \label{fig:shades}
\end{figure}

To this end, we measure how much can we extend simultaneously,
with the same string to the left (respectively, to the right), the occurrences of $v$ corresponding to these $w_i$'s. This will give us ranges of the same length, around each of the $v$'s, that contain all possible $w_i$'s. We follow a similar strategy for the $w_j$'s corresponding to $\lvec{x}$ in $p$ (see the details below). Now, all it remains is to see whether we can glue together some occurrences of $w_1,w_2,\ldots,w_r$ from the respective ranges, by identifying between them exactly the strings $s_1,s_2,\ldots,s_r$. This is be done efficiently using the arrays storing the occurrences of the $s_i$'s, and standard bitwise operations.
Let us formalize this explanation.

For $z \in [1..r)$, denote $q_z = h_1 + |s_2s_3\cdots s_z| + (z{-}1)\beta$. Denote by $Z$ (resp., $\lvec{Z}$) the set of all $z\in [1..r)$ such that $x_z = x$ (resp., $x_z = \lvec{x}$). If there is an instance $t[i..j] = s_1w_1s_2w_2\cdots s_{r-1}w_{r-1}s_r$ of $p$ such that $|w_1| = \cdots = |w_{r-1}| = \beta$ and $i + |s_1| \le h_1 < h_2 < i + |s_1w_1|$, then, for any $z, z' \in Z$ (resp., $z, z' \in \lvec{Z}$), $t[q_z..q_z{+}|v|{-}1] = t[q_{z'}..q_{z'}{+}|v|{-}1]$. We check these equalities in $O(r)$ time using the $\lcp$ structure. Suppose this checking succeeds. There might exist many corresponding instances of $p$ as in Fig.~\ref{fig:shades}.

We can immediately calculate the numbers $b_{\ell} = \min\{\rlcp(q_z{-}1, q_{z'}{-}1) \colon$ $(z, z') \in (Z\times Z) \cup (\lvec{Z}\times\lvec{Z})\}$ and $b_{r} = \min\{\lcp(q_z{+}|v|, q_{z'}{+}|v|)\colon$ $(z, z') \in (Z\times Z)\cup(\lvec{Z}\times\lvec{Z})\}$ in $O(r)$ time. Assume that $t[i..j] = s_1w_1s_2w_2\cdots s_{r-1}w_{r-1}s_r$ is an instance of $p$ with $|w_1| = \cdots = |w_{r-1}| = \beta$ and $i + |s_1| \le h_1 < h_2 < i + |s_1w_1|$. By the definition of $b_{\ell}$ and $b_r$, we then necessarily have $q_z - \delta \ge q_z - b_{\ell}$ and $q_z - \delta + \beta \le q_z + |v| + b_{r}$ for all $z \in [1..r)$, where $\delta = h_1 - (i + |s_1|)$.

Thus, the next segments are non-empty (see Fig.~\ref{fig:shades}):\\
\indent $S_z = [q_z - |s_z| - b_{\ell}\ ..\ q_{z-1} + |v| + b_{r}]\ \cap\ [q_{z-1} + |v|\ ..\ q_z - |s_z|]\text{ for }z \in (1..r)$,\\
\indent $S_1 = [q_1 - |s_1| - b_{\ell}\ ..\ q_1 - |s_1|]\ \cap\ [q_1 + |v| - |s_1| - \beta\ ..\ q_1 - |s_1|]$,\\
\indent $S_r = [q_{r-1} + |v|\ ..\ q_{r-1} + |v| + b_{r}]\ \cap\ [q_{r-1} + |v|\ ..\ q_{r-1} + \beta].$

Further, if such instance $t[i..j]$ exists, then there is a sequence of positions $\{i_z\}_{z=1}^r$ such that $i_z \in S_z$, $D_z[i_z] = 1$ for $z \in [1..r]$ and $i_{z+1} - i_z = |s_z| + \beta$ for $z \in [1..r)$ (namely, $i_1 = i$). If $x_1 = \cdots = x_{r-1}$, then the converse is also true: if a sequence $\{i_z\}_{z=1}^r$ satisfies all these conditions, then $t[i_1..i_r{+}|s_r|{-}1] = s_1ws_2w\cdots s_{r-1}ws_r$, where $|w| = \beta$ and $i + |s_1| \le h_1 < h_2 < i + |s_1| + \beta$. The bit arrays $\{D_z\}_{z=1}^r$ help us to find all such sequences.

Let  $D'_1{=}D_1[q_1{+}|v|{-}|s_1|{-}\beta .. q_1{-}|s_1|]$, $D'_r{=}D_r[q_{r-1}{+}|v| .. q_{r-1}{+}\beta]$ and $D'_z = D_z[q_{z-1}{+}|v|..q_z{-}|s_z|]$ for $z \in (1..r)$. For each $z \in [1..r]$, we clear in the array $D'_z$ all bits corresponding to the regions that are not covered by the segment $S_z$ and then perform the bitwise ``and'' of $D'_1,\ldots,D'_r$; thus, we obtain a bit array $D[0..\beta{-}|v|]$ (see Fig.~\ref{fig:shades}). If $x_1 = \cdots = x_{r-1}$, then, for any $i \in [0..\beta{-}|v|]$, we have $D[i] = 1$ iff there is a string $s_1ws_2w\cdots s_{r-1}ws_r$ starting at $i' = h_2 - \beta - |s_1| + i + 1$ such that $|w| = \beta$ and $i' + |s_1| \le h_1 < h_2 < i' + |s_1w|$. Obviously, one can put $occ[h_2{-}\beta{-}|s_1|{+}1..h_1{-}|s_1|] = D[0..\beta{-}|v|]$. Since the length of each of the arrays $D'_1, \ldots, D'_r$ does not exceed $\beta$, all these calculations can be done in $O(r{+}\frac{r\beta}{\log n})$ time by standard bitwise operations on the $\Theta(\log n)$-bit machine words.

If $p$ contains both $x$ and $\lvec{x}$, it is not clear how to check whether the substitutions of $x$ and $\lvec{x}$ corresponding to a given $D[i] = 1$ respect each other. The case when $p$ contains both $x$ and $\lvec{x}$ turns out to be much more difficult; see Appendix.
\qed
\end{proof}

\section{Periodic Anchor Substring $v$}\label{SectPeriodic}

In this section we suppose $v$ is periodic. Recall that $v$ starts at $q_1$ and we also know its length. By Lemma~\ref{SubstringRun}, we find in $O(1)$ time the minimal period $d$ of $v$ and a run $t[i'..j']$ with period $d$ containing $v$ (i.e., $i' \le q_1 < q_1 + |v| - 1 \le j'$).

Just like before, we are searching for instances $t[i..j] = s_1w_1\cdots s_{r-1}w_{r-1}s_r$ of $p$ such that $\frac{3}2|v| < |w_1| \le 2|v|$ and $v$ occurs in $w_1$, so at least $|s_1|$ symbols away from $i$ (in other words, $i + |s_1| \le q_1 < q_1 + |v| \le i + |s_1w_1|$). Let us assume that $t[i..j]$ is such an instance. Then, either $w_1$ has period $d$ or one of the strings $v' = t[q_1 .. j'{+}1]$ or $v'' = t[i'{-}1 .. q_1{+}|v|{-}1]$ is a substring of $w_1$ (that is, the run containing $v$ ends or, respectively, starts strictly inside $w_1$).

Suppose first that $w_1$ contains $v'$ as a substring (the case of $v''$ is similar); note that $v'$ is the suffix of the run $t[i'..j']$ starting at position $q_1$, to which a letter that breaks the period was added. One can show that, since the minimal period of $t[q_1 .. j']$ is $d$, $2d \le j' - q_1 + 1$, and $t[j'{+}1] \ne t[j'{+}1{-}d]$, the string $v'$ is not periodic. Hence, $v'$ can be processed in the same way as $v$ in Section~\ref{SectNonPeriodic}, and get the instances of $p$ that occur around it. A similar conclusion is reached when $w_1$ contains $v''$, so we assume in the following that $w_1$ is periodic.

Suppose that $w_1$ has period $d$. Periodic substitutions of $x$ (such as $w_1$) can produce a lot of instances of $p$: e.g., $a^n$ contains $\Theta(n^2)$ instances of $xx$. However, it turns out that when such multiple instances really occur, they have a uniform structure that can be compactly encoded and appear only when all substitutions of $x$ and $\lvec{x}$ lie either within one or two runs. Before the discussion of this case, let us first consider the case when three or more runs contain $w_1,\ldots,w_{r-1}$.
Due to space constraints, some proofs are moved to Appendix.

\paragraph{Three and more runs.}
Let $t[i..j]$ be an instance of $p$ with a substitution of $x_1=x$ denoted by $w_1=w$ and such that $w$ has period $d$. Moreover, for our chosen $v$ starting at position $q_1$, we still have $\frac{3}2|v| < |w| \le 2|v|$ and $v$ occurs inside $w_1$ (i.e., $i + |s_1| \le q_1 < q_1 + |v| \le i + |s_1w|$). Since $|v| \ge 2d$, we have $|w| \ge \frac{3}2|v| \ge 3d$. Clearly, each substitution of $x$ or $\lvec{x}$ in $t[i..j]$ is contained in some run with period $d$ (some of these runs may coincide). It turns out that if all substitutions of $x$ and $\lvec{x}$ in $t[i..j]$ are contained in at least three distinct runs with period $d$, then there are only constantly many possibilities to choose the length $|w|$, and these possibilities can be efficiently found and then processed by Lemma~\ref{FixedLen} to find the instances of the pattern. To begin with, let us introduce several  lemmas; in their statements $w$ and $s$ are strings (extensions for reversals are given in Appendix).

\begin{lemma}
Let $ws$ be a substring of $t$ such that $w$ has period $d$, $|w| \ge 3d$, and $ws$ does not have period $d$. Let $t[i..j]$ be a run with period $d$ containing $w$ and let $h$ be the starting position of $s$. Then, either $h = j - |\pref_d(s)| + 1$ or $h \in (j{+}1{-}d..j{+}1]$.\label{wsw}
\end{lemma}
\begin{proof}
Suppose that $h \le j + 1 - d$. Then, $|\pref_d(s)| \ge j - h + 1 \ge d$. Thus, since $t[j{+}1] \ne t[j{+}1{-}d]$, $|\pref_d(s)|$ must be equal to $j - h + 1$ and hence $h = j - |\pref_d(s)| + 1$.\qed
\end{proof}

\begin{lemma}
Let $wsw$ (resp., $\lvec{w}sw$)
be a substring of $t$ such that $w$ has period $d$, $|w| \ge 3d$, and $wsw$ (resp., $\lvec{w}sw$)
does not have period $d$. Let $t[i..j]$ be a run with period $d$ containing the first occurrence of $w$ (resp., $\lvec{w}$)
in $wsw$ (resp., $\lvec{w}sw$).
Denote by $h$ the starting position of~$s$. Then, we have $h = j - |\pref_d(s)| + 1$ or $h \in (j{+}1{-}d..j{+}1]$ or $h \in (j{-}|s|{-}d..j{-}|s|]$.\label{wsw2}
\end{lemma}
\begin{proof}
If $h + |s| > j$, then, by Lemma~\ref{wsw}, either $h = j - |\pref_d(s)| + 1$ or $h \in (j{+}1{-}d..j{+}1]$. Suppose that $h + |s| \le j$. Let $t[i'..j']$ be a run with period $d$ containing the last occurrence of $w$ in $wsw$ (resp., $\lvec{w}sw$). Clearly, $i' \le h + |s|$. Hence, since $t[i..j]$ and $t[i'..j']$ cannot overlap on $d$ letters, we obtain $j - d + 1 < h + |s|$. Therefore, $h \in (j{-}|s|{-}d..j{-}|s|]$.\qed
\end{proof}

As the string $w$ is periodic, but the whole image of $p$ is not (it extends over three or more runs), some of the strings $s_z$ must break the period induced by $w$. If we can identify the $s_z$'s which break the period, Lemmas~\ref{wsw} and~\ref{wsw2} allow us to locate their occurrences which, together with the $v$ we considered, might lead to finding corresponding instances of $p$. The next lemma formalizes these ideas (its proof, especially for the patterns containing both $x$ and $\lvec{x}$, is rather non-trivial and uses results from~\cite{GalilSeiferas2,GawrychowskiLewensteinNicholson,KosolobovPalk,Manacher,RubinchikShur}; see Appendix.

\begin{lemma}
Let $v = t[h_1..h_2]$ be a string with the minimal period $d \le \frac{|v|}2$. Given $z, z'$ such that $1 < z < z' < r$, we can find all instances $t[i..j] = s_1w_1s_2w_2\cdots s_{r-1}w_{r-1}s_r$ of $p$ such that $\frac{3}2|v| < |w_1| \le 2|v|$, $v$ is contained in $w_1$,
$w_1s_2w_2\cdots s_{z-1}w_{z-1}$ and $w_zs_{z+1}w_{z+1}\cdots s_{z'-1}w_{z'-1}$ both have period $d$, and $w_{z-1}s_zw_z$ and $w_{z'-1}s_{z'}w_{z'}$ both do not have period $d$, in $O(r + \frac{r|v|}{\log n})$ time. To this we add $O(\frac{\log n}{\log\log n})$ time if $\frac{\log n}{16\log\log n} \le |v| \le \log n$.\label{SeparateLemma}
\end{lemma}

It remains to explain how to identify the $s_z$'s that break the period inside the instances of $p$, and show that their number is $O(1)$.
Let $Z$ (resp., $Z'$, $Z''$) be the set of all numbers $z \in (1..r)$ such that $x_{z-1} = x_z$ (resp., $\lvec{x}_{z-1} = x_z = \lvec{x}$, $x_{z-1} = \lvec{x}_z = \lvec{x}$). By Lemma~\ref{PrimitiveVV}, as $w_z\in\{w,\lvec{w}\}$ for $z\in [1..r]$, the next lemma~follows:

\begin{lemma}
For any numbers $z_1, z_2 \in Z$ (resp.,~$ Z'$, $ Z''$), if the strings $w_{z_1-1}s_{z_1}w_{z_1}$ and $w_{z_2-1}s_{z_2}w_{z_2}$ both have period $d$, then the next properties hold:
\begin{equation}
\begin{array}{l}
|s_{z_1}| \equiv |s_{z_2}| \pmod{d},\ \ s_{z_1}\text{ and }s_{z_2}\text{ both have period }d,\\
\text{one of }s_{z_1}\text{ and }s_{z_2}\text{ (}s_{z_1}\text{ and }\lvec{s}_{z_2}\text{ if }x_{z_1} \ne x_{z_2}\text{) is a prefix of another.}
\end{array}\label{eq:propwsw}
\end{equation}
In the following sense, the converse is also true: if $|s_{z_1}| \ge d$, $w_{z_1-1}s_{z_1}w_{z_1}$ has period $d$, and $z_1$ and $z_2$ satisfy~\eqref{eq:propwsw}, then $w_{z_2-1}s_{z_2}w_{z_2}$ necessarily has period $d$.\label{SepCond}
\end{lemma}

We call a pair of numbers $(z,z')$ such that $z \le z'$ and $z,z' \in Z$ a \emph{separation in $Z$} if all numbers $z_1, z_2 \in ((1..z)\cup (z..z'))\cap Z$ satisfy~\eqref{eq:propwsw} and all numbers $z_1 \in ((1..z) \cup (z..z')) \cap Z$ and $z_2 \in \{z,z'\}$ either do not satisfy~\eqref{eq:propwsw} or satisfy $|s_{z_1}| < d \le |s_{z_2}|$; separations in $Z'$ and $Z''$ are defined analogously. Informally, a pair $(z,z')$ is a separation in $Z$ (resp., $Z'$, $Z''$) if $w_1s_2\cdots s_{z-1}w_{z-1}$ and $w_zs_{z+1}\cdots s_{z'-1}w_{z'-1}$ both have period $d$, and $w_{z-1}s_zw_z$ and $w_{z'-1}s_{z'}w_{z'}$ both do not have period~$d$.

In other words, such a pair indicates exactly the first two $s_z$'s where the period breaks in an instance of $p$.
Accordingly, if we will apply Lemma~\ref{SeparateLemma} for all pairs $(z,z')$ such that $z$ and $z'$ occur in some separations in $Z$ or $Z'$ or $Z''$, then we will find all instances $t[i..j] = s_1w_1s_2w_2\cdots w_{r-1}s_r$ of $p$ such that $w_1,\ldots,w_{r-1}$ lie in at least three distinct runs with period $d$, $\frac{3}2|v| < |w_1| \le 2|v|$, and $v$ occurs in $w_1$. So, it suffices to show that there are at most $O(1)$ possible separations in $Z$ (resp., $Z'$, $Z''$) and, to reach the complexity announced in the \emph{General Strategy} section, all such separations can be found in $O(r)$ time.

We describe how to find all separations in $Z$ (the cases of $Z'$, $Z''$ are similar). Clearly, if $(z,z')$ is a separation, then $(z,z)$ is also a (degenerate) separation. We find all separations $(z,z)\in Z$ applying the following general  lemma with $Z_0 = Z$.
\begin{lemma}
For any subset $Z_0 \subseteq Z$ (resp., $Z_0 \subseteq Z'$, $Z_0 \subseteq Z''$), there are at most three numbers $z \in Z_0$ satisfying the following property~\eqref{eq:subz}:%
\begin{equation}
\begin{array}{l}
\text{any }z_1, z_2 \in (1..z)\cap Z_0\text{ satisfy~\eqref{eq:propwsw}},\\
\text{any }z_1 \in (1..z)\cap Z_0, z_2 = z\text{ either do not satisfy~\eqref{eq:propwsw} or }|s_{z_1}| < d \le |s_{z_2}|.\label{eq:subz}
\end{array}
\end{equation}
All such $z$ can be found in $O(r)$ time.\label{subz}
\end{lemma}
\begin{proof}
Let $z' = \min Z_0$. Clearly, $z = z'$ satisfies~\eqref{eq:subz}. Using the $\lcp$ structure on the string $p\lvec{p}$, we find in $O(r)$ time  the smallest number $z'' \in Z_0$ such that any $z_1, z_2 \in [z'..z'') \cap Z_0$ satisfy~\eqref{eq:propwsw} and some $z_1, z_2 \in [z'..z''] \cap Z_0$ do not satisfy~\eqref{eq:propwsw}; assume $z'' = +\infty$ if there is no such $z''$. Obviously, if $z'' \ne +\infty$, then $z = z''$ satisfies~\eqref{eq:subz}. Any $z \in (z''..{+}\infty)\cap Z_0$ does not satisfy~\eqref{eq:subz} because in this case $z'' \ne +\infty$ and some $z_1, z_2 \in [z'..z'']\cap Z_0$ do not satisfy~\eqref{eq:propwsw}. In $O(r)$ time we find the minimal $z''' \in [z'..z'')\cap Z_0$ such that $|s_{z'''}| \ge d$; assume $z''' = z''$ if there is no such $z'''$. By the definition, we have $s_{z_1} = s_{z_2}$ and $|s_{z_1}| = |s_{z_2}| < d$ for any $z_1, z_2 \in [z'..z''') \cap Z_0$. Therefore, any $z \in (z'..z''') \cap Z_0$ does not satisfy~\eqref{eq:subz}. Further, any $z \in (z'''..z'') \cap Z_0$ does not satisfy~\eqref{eq:subz} since in this case $z_1 = z'''$ and $z_2 = z$ satisfy~\eqref{eq:propwsw} and $|s_{z_1}| \ge d$, which contradicts to~\eqref{eq:subz}. Finally, if $z''' \ne +\infty$, then $z = z'''$ obviously satisfies~\eqref{eq:subz}. So, $z'$, $z''$, $z'''$ are the only possible numbers in $Z_0$ that can satisfy~\eqref{eq:subz}.\qed
\end{proof}

Finally, for each separation $(z,z)\in Z$ we have found, we apply Lemma~\ref{subz} with $Z_0 = Z \setminus\{z\}$ and obtain all separations in $Z$ of the form $(z,z')$ for $z' > z$. Employing Lemma~\ref{subz} at most three times, we obtain at most~9 new separations in total, in $O(r)$ total time, and, besides the at most three $(z,z)$ separations we initially had, no other separations exist. So, there are at most $12$ separations in $Z$ and they can be found in $O(r)$ time. Lemma \ref{SeparateLemma} can be now employed to conclude the identification of the instances of $p$ extending over at least three runs.

\paragraph{In-a-run instances of $p$.} This case requires a different approach.
More precisely, we process each run $t[i'..j']$ (only once) with period $d$ in order to find all instances $t[i..j]$ of $p$ satisfying the following properties (denoted altogether as~(\ref{eq:mainprop})):
\begin{equation}
\begin{array}{l}
t[i..j]\text{ is an instance of }p\text{ with substitutions of }x\text{ and }\lvec{x}\text{ of length }{\ge}3d,\\
t[i + |s_1|..j - |s_r|]\text{ is a substring of }t[i'..j'].
\end{array}\label{eq:mainprop}
\end{equation}
So, in this case, we no longer try to extend the string $v$ that anchors the occurrence of $w_1$, but have a more global approach to finding the instances of the pattern.

To begin with, since $t[i'..j']$ has period $d$, we obtain the following lemma.
\begin{lemma}
Let $t[i..j]$ be a string satisfying~\eqref{eq:mainprop} such that $i' \le i$. Then $t[i{+}(r{-}1)d..j]$ is an instance of $p$ and, if $i - (r-1)d \ge i'$, $t[i{-}(r{-}1)d..j]$ is also an instance of $p$.\label{AllInRunOccurrences}
\end{lemma}

Let $t[i..j]$ satisfy~\eqref{eq:mainprop} and $w$ be a substitution of $x$ in $t[i..j]$. Recall that $x_1=x$ and $r\geq 3$. We try to get some information on $|w|$, the length of the substitution of $x$. Suppose that $p \ne s_1xs_2\lvec{x}s_3$ (the case $p = s_1xs_2\lvec{x}s_3$ is considered in Appendix). Then, either there is $z \in (1..r)$ such that $x_{z-1} = x_z$ or there are $z',z'' \in (1..r)$ such that $x_{z'-1}s_{z'}x_{z'} = \lvec{x}s_{z'}x$ and $x_{z''-1}s_{z''}x_{z''} = xs_{z''}\lvec{x}$. Accordingly, we can compute the number $|w| \bmod d$ as follows.

\begin{lemma}
Let $t[i..j]$ satisfy~\eqref{eq:mainprop} and $w$ be a substitution of $x$ in $t[i..j]$. If, for some $z \in (1..r)$, $x_{z-1} = x_z$, then $|w| \equiv -|s_{z}| \pmod{d}$; if, for some $z', z'' \in (1..r)$, $x_{z'-1}s_{z'}x_{z'} = \lvec{x}s_{z'}x$ and $x_{z''-1}s_{z''}x_{z''} = xs_{z''}\lvec{x}$, then either $|w| \equiv \frac{d - |s_{z''}| - |s_{z'}|}{2} \pmod{d}$ or $|w| \equiv \frac{-|s_{z''}| - |s_{z'}|}{2} \pmod{d}$.\label{wmodd}
\end{lemma}
\begin{proof}
Suppose that $x_{z-1} = x_z$. Since, by Lemma~\ref{PrimitiveVV}, the distance between any two occurrences of $w$ (or $\lvec{w}$) in $t[i'..j']$ is a multiple of $d$, we have $|w| \equiv -|s_z| \pmod{d}$.

Suppose that $x_{z'-1}s_{z'}x_{z'} = \lvec{x}s_{z'}x$ and $x_{z''-1}s_{z''}x_{z''} = xs_{z''}\lvec{x}$. Since $w$ and $\lvec{w}$ both are substrings of $t[i'..j']$ and $|w| \ge 3d$, it follows from Lemma~\ref{Palsuv} that there are palindromes $u$ and $v$ such that $|uv| = d$, $v \ne \epsilon$, and $\lvec{w}$ is a prefix of the infinite string $(vu)^{\infty}$. Since $ws_{z''}\lvec{w}$ is a substring of $t[i'..j']$ and the strings $vu$ and $uv$ are primitive, it follows from Lemma~\ref{PrimitiveVV} that $s_{z''} = u(vu)^{k'}$ for an integer $k'$ and hence $|u| = |s_{z''}| \bmod d$, $|v| = d - |u|$. Similarly, since $\lvec{w}s_{z'}w$ is a substring of $t[i'..j']$, we have $\lvec{w}s_{z'}w = (vu)^{k'}v$ for an integer $k'$ and therefore $2|w| \equiv |v| - |s_{z'}| \pmod{d}$. Thus, either $|w| \equiv \frac{|v| - |s_{z'}|}{2} \pmod{d}$ or $|w| \equiv \frac{d + |v| - |s_{z'}|}{2} \pmod{d}$. Since $|v| = (-|s_{z''}|) \bmod d$, we obtain either $|w| \equiv \frac{d - |s_{z''}| - |s_{z'}|}{2} \pmod{d}$ or $|w| \equiv \frac{-|s_{z''}| - |s_{z'}|}{2} \pmod{d}$.\qed
\end{proof}

We now fix the possible ends of the instances $t[i..j]$ of the pattern $p$, with respect to $t[i'..j']$. Consider the segments $\{(j'{+}1{-}bd..j'{+}1{-}(b{-}1)d]\}_{b=1}^f$, where $f$ is the maximal integer such that $j'{+}1{-}fd \ge i'$ (i.e., $f$ is exponent of the period in the run $t[i'..j']$).  For each $b \in [1..f]$, we can find in $O(r + \frac{rd}{\log n})$ time, using Lemma~\ref{GeneralInRunLemma}, all strings $t[i..j]$ satisfying~\eqref{eq:mainprop} such that $j{-}|s_r|{+}1 \in (j'{+}1{-}bd..j'{+}1{-}(b{-}1)d]$ (so with $s_1w_1\ldots s_{r-1}w_{r-1}$ ending in the respective segment); the parameter $\delta$ in this lemma is chosen according to Lemma~\ref{wmodd} (see below). Adding all up, this enables us to find all instances of $p$ satisfying~\eqref{eq:mainprop} in $O(r(\frac{j' - i' + 1}{d} + \frac{j' - i' + 1}{\log n}))$ time. Since $j' - i' + 1 \ge 3d$, it follows from~\cite{KolpakovKucherov} and Lemma~\ref{sum_runs} that the sum of the values $\frac{j' - i' + 1}d + \frac{j' - i' + 1}{\log n}$ over all such runs $t[i'..j']$ is $O(n)$; hence the total time needed to find these instances of the pattern is~$O(rn)$.

Technically, our strategy is given in the following lemma (which does not cover the case $p = s_1xs_2\lvec{x}s_3$, which is present in Appendix). In this lemma, $\delta \in [0..d)$ is one of the possible values of $|w| \bmod d$, as obtained in Lemma~\ref{wmodd}: if $x_{z-1} = x_{z}$ for some $z \in (1..r)$, we use only one value $\delta = -|s_z|\bmod d$; otherwise, we use two values of $\delta$ described in Lemma~\ref{GeneralInRunLemma} (the special case $p = s_1xs_2\lvec{x}s_3$ is considered separately in Appendix). For each thus computed $\delta$, we process each segment $[b_1..b_2] = (j'{+}1{-}bd..j'{+}1{-}(b{-}1)d]$ and get a compact representation (in the bit arrays $E$, $F$) of the instances $s_1w_1\ldots s_{r-1}w_{r-1}s_r$ of $p$ such that $s_1w_1\ldots s_{r-1}w_{r-1}$ ends in the respective segment and $\delta = |w| \mod d$. The proof of Lemma~\ref{GeneralInRunLemma} is moved to Appendix.

\begin{lemma}
Let $p \ne s_1xs_2\lvec{x}s_3$, $r \ge 3$, and $\delta \le d$.
Given a run $t[i'..j']$ with period $d$ and a segment $[b_1..b_2] \subset [i'..j'{+}1]$ of length $d$, we can compute in $O(r{+}\frac{rd}{\log n})$ time the numbers $d', d'', h', h'', a', a''$ and bit arrays $E[b_1..b_2]$, $F[b_1..b_2]$ such that:
\begin{enumerate}
\item for any $h \in [b_1..h']$ (resp., $h\in (h'..b_2]$), we have $E[h] = 1$ iff the strings $t[h{-}|s_1s_2\cdots s_{r-1}|{-}(r{-}1)(\delta{+}cd)..h{+}|s_r|{-}1]$ for all $c \in [0..d']$ (resp., for all $c \in [0..d'']$) are instances of $p$ and $h - |s_1s_2\cdots s_{r-1}| - (r - 1)(\delta + cd) \ge i'$;
\item for any $h \in [b_1..h'']$ (resp., $h \in (h''..b_2]$), we have $F[h] = 1$ iff the string $t[h{-}|s_1s_2\cdots s_{r-1}|{-}(r{-}1)(\delta{+}ad)..h{+}|s_r|{-}1]$, where $a = a'$ (resp., $a = a''$), is an instance of $p$ and $h - |s_2s_3\cdots s_{r-1}| - (r - 1)(\delta + ad) \ge i'$.
\end{enumerate}
In addition, we find at most one instance $t[i_0..j_0] = s_1w_1s_2w_2\cdots w_{r-1}s_r$ of $p$ satisfying~\eqref{eq:mainprop} and such that $j_0 - |s_r| + 1 \in [b_1..b_2]$, $|w_1| \equiv \delta \pmod{d}$, and it is guaranteed that if a string $t[i..j] = s_1w_1s_2w_2\cdots w_{r-1}s_r$ satisfies~\eqref{eq:mainprop}, $j - |s_r| + 1 \in [b_1..b_2]$, and $|w_1| \equiv \delta \pmod{d}$, then either $t[i..j]$ is encoded in one of the arrays $E$, $F$ or $i = i_0$ and $j = j_0$.\label{GeneralInRunLemma}
\end{lemma}

The only case left is of in-two-runs instances of $p$. To solve this case we combine (in a rather technical way) the ideas of the previous cases. Instances of $p$ extending over two runs are determined by separators (as the period breaks once inside these instances), but the prefix and suffix of each instance, occurring before, resp. after, the separator can be extended just as in the case of instances occurring inside a single run, discussed above. The details are given in Appendix.

\bibliographystyle{splncs03}
\bibliography{pat}

\newpage
\section*{Appendix}

\subsection*{To Section~\ref{SectNonPeriodic}}

\paragraph{The continuation of the discussion before Lemma~\ref{FixedLen}.}

We have $v = t[q_1..q_1{+}|v|{-}1]$. Let $q_2 \in [q_1 + |vs_2| .. q_1 + |vs_2vv|]$ be the starting position of an occurrence of $\lvec{v}$. If $p = s_1xs_2\lvec{x}s_3$ or $p = s_1\lvec{x}s_2xs_3$, then we apply the following lemma to find all instances $t[i..j] = s_1ws_2\lvec{w}s_3$ of $p$ such that $\frac{3}2|v| < |w| \le 2|v|$, $i + |s_1| \le q_1 < q_1 + |v| \le i + |s_1w|$, and $q_1 + |v| - (i + |s_1|) = (i + |s_1ws_2w|) - q_2$; the latter equality guarantees that the string $t[q_2..q_2{+}|v|{-}1]$ in such instance is a reversal of $t[q_1..q_1{+}|v|{-}1]$ produced by the substitution $\lvec{w}$ (see Fig.~\ref{fig:xslvecx}).
\begin{lemma}
Let $p = s_1xs_2\lvec{x}s_3$ or $p = s_1\lvec{x}s_2xs_3$. Given a substring $t[h_1..h_2] = v$ and a position $q > h_2$ such that $t[q..q{+}|v|{-}1] = \lvec{v}$, we can compute in $O(1 + \frac{|v|}{\log n})$ time a bit array $occ[h_1{-}|s_1v|..h_1{-}|s_1|]$ such that, for any $i$, $occ[i] = 1$ iff $t[i..n]$ has a prefix $s_1ws_2\lvec{w}s_3$ such that $h_2 - (i + |s_1|) = (i + |s_1ws_2w| - 1) - q$ and $i + |s_1| \le h_1 < h_2 < i + |s_1w|$ (see Fig.~\ref{fig:xslvecx}).\label{PalPattern}
\end{lemma}
\begin{proof}
We first test whether $(h_2 + 1 + q - |s_2|)/2$ is integer and $D_2[(h_2 + 1 + q - |s_2|)/2] = 1$ to check that $s_2$ occurs precisely between the substitutions of $x$ and $\lvec{x}$ (see Fig.~\ref{fig:xslvecx}). Using the $\lcp$ structure for the string $t\lvec{t}$, we check in $O(1)$ time  that $t[h_2{+}1] = t[q{-}1], t[h_2{+}2] = t[q{-}2], \ldots$ and find the length $b$ of the longest common prefix of $\lrange{t[1..h_1{-}1]}$ and $t[q{+}|v|..n]$. Then, for each $i \in [h_1 - |s_1| - \min\{b,|v|\} .. h_1 - |s_1|]$, we have $occ[i] = 1$ iff $D_1[i] = 1$ and $D_3[q + h_2 - (i + |s_1|) + 1] = 1$; for all other $i \in [h_1 - |s_1v| .. h_1 - |s_1|]$, we have $occ[i] = 0$. Thus, we obtain $occ$ performing the bitwise ``and'' of the corresponding subarray of $D_1$ and the corresponding reversed subarray of $D_3$. The length of both these subarrays is bounded by $|v|$. To obtain the reversed subarray efficiently, we utilize a precomputed table of size $O(2^{\frac{\log n}{2}}) = O(\sqrt{n})$ that allows us to reverse the order of bits in one $\Theta(\log n)$-bit machine word in $O(1)$ time. Thus, the running time of this algorithm is $O(\frac{|v|}{\log n})$.\qed
\begin{figure}[htb]
\includegraphics[scale=0.2]{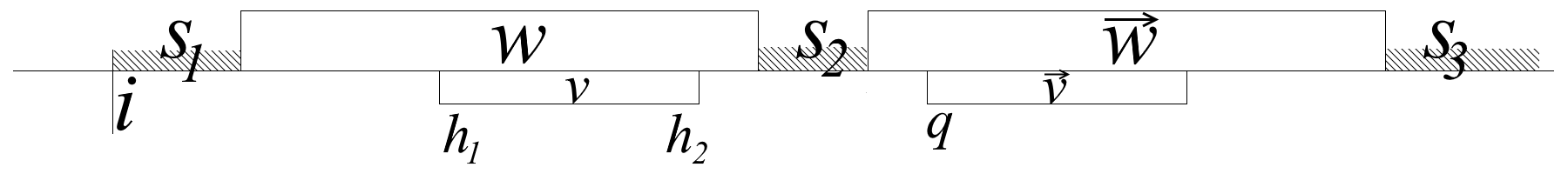}
\caption{\small An instance of $p = s_1xs_2\lvec{x}s_3$ in the proof of Lemma~\ref{PalPattern}.}
\label{fig:xslvecx}
\end{figure}
\end{proof}

Suppose that $p$ starts with $s_1xs_2\lvec{x}$ but is not equal to $s_1xs_2\lvec{x}s_3$ (the case when $p$ starts with $s_1\lvec{x}s_2x$ is analogous). We are to compute all possible lengths of the substitutions of $x$ from the range $(\frac{3}2|v|..2|v|]$ that could produce the found occurrences of $v$ and $\lvec{v}$ starting at positions $q_1$ and $q_2$, respectively. Since $v$ is not periodic, there are at most four occurrences of $v$ [or~$\lvec{v}$] starting in the segment $[q_2 + |vs_3|..q_2 + |vs_3vv|]$. We find all these occurrences with the aid of Lemma~\ref{vFind} putting $\lambda = |s_3|$. If $x_3 = x$ (resp.,~$x_3 = \lvec{x}$), then, given the starting position $q_3$ of an occurrence of $v$ (resp.,~$\lvec{v}$) such that $q_3 \in [q_2 + |vs_3| .. q_2 + |vs_3vv|]$, the number $\beta = (q_3 - q_1 - |s_2s_3|) / 2$ (resp.,~$\beta = q_3 - q_2 - |s_3|$) is equal to the length of the corresponding substitution of $x$ that could produce the found occurrences of $v$ and $\lvec{v}$ starting at $q_1$, $q_2$, and $q_3$. Thus, we obtain a constant number of possible lengths for substitutions of $x$~and, for each of the found lengths, we apply Lemma~\ref{FixedLen}.

To discuss the full version of the proof of Lemma~\ref{FixedLen}, we need an additional tool. A \emph{Lyndon root} (\emph{reversed Lyndon root}) of a run $t[i..j]$ of period $d$ is its lexicographically smallest substring $t[\ell{+}1..\ell{+}d]$ (resp., $\lrange{t[\ell{+}1..\ell{+}d]}$) with~$\ell \in [i{-}1..j{-}d]$. The following result is Lemma~1 from~\cite{CrochemoreEtAl}.

\begin{lemma}
The leftmost (reversed) Lyndon root of any run in $t$ can be found in $O(1)$ time assuming $O(n)$ time preprocessing.\label{LyndonRoots}
\end{lemma}

Now we are ready to present the proof of Lemma~\ref{FixedLen}, which is full of non-trivial technical details.

\begin{proof}[Lemma~\ref{FixedLen}, continuation]
\emph{The case when $p$ contains both $x$ and $\lvec{x}$.}
For brevity, denote $s(i) = h_2 - \beta - |s_1| + i + 1$. Our aim is to ``filter'' the bit array $D[0..\beta{-}|v|]$ so that, for any $i$, it will be guaranteed that $D[i] = 1$ iff there is an instance $s_1w_1s_2w_2\cdots s_{r-1}w_{r-1}s_r$ of $p$ starting at $s(i)$ and such that $|w_1| = \cdots = |w_{r-1}| = \beta$ and $s(i) + |s_1| \le h_1 < h_2 < s(i) + |s_1| + \beta$. Without loss of generality, suppose that $p$ has a prefix $s_1x$ (the case of the prefix $s_1\lvec{x}$ is symmetrical). For each $h\in (0..8]$, we ``filter'' the subarray $D[(y - 1)\lceil\frac{\beta}{8}\rceil..y\lceil\frac{\beta}{8}\rceil{-}1]$ in $O(r + \frac{\beta}{\log n})$ time (assuming $D[i] = 0$ for $i > \beta - |v|$ so that these subarrays are well defined); hence, the overall running time of our filtration algorithm is $O(r + \frac{\beta}{\log n})$.

Fix $y \in (0..8]$. Suppose that there are two positions $i, i' \in [(y - 1)\lceil\frac{\beta}{8}\rceil..y\lceil\frac{\beta}{8}\rceil)$ such that $i < i'$ and there are two instances of $p$ starting at positions $s(i)$ and $s(i')$, respectively, such that the lengths of the substitutions of $x$ in both these instances are equal to $\beta$; denote by $w$ and $w'$ the corresponding substitutions of $x$ in these instances. Clearly $D[i] = D[i'] = 1$. It turns out that in this case $w$ and $w'$ both are periodic and, relying on this fact, we will deduce some regularities in the distribution of positions $i'' \in [(y - 1)\lceil\frac{\beta}{8}\rceil..y\lceil\frac{\beta}{8}\rceil)$ such that $t[s(i'') .. s(i'') + |s_1s_2\cdots s_r| + (r - 1)\beta - 1]$ is an instance of $p$; so, it will suffice to assign $D[i''] = 0$ for all ``non-regular'' positions $i''$. Let us describe precisely the nature of these regularities.

Denote $\gamma = i' - i$. Fix the minimal number $z \in (1..r)$ such that $x_z = \lvec{x}$. Since $\gamma \le \frac{\beta}{8}$, it follows from Fig.~\ref{fig:respect} that $w$ and $w'$ both have period $2\gamma \le \frac{\beta}{4}$.

\begin{figure}[htb]
\includegraphics[scale=0.3]{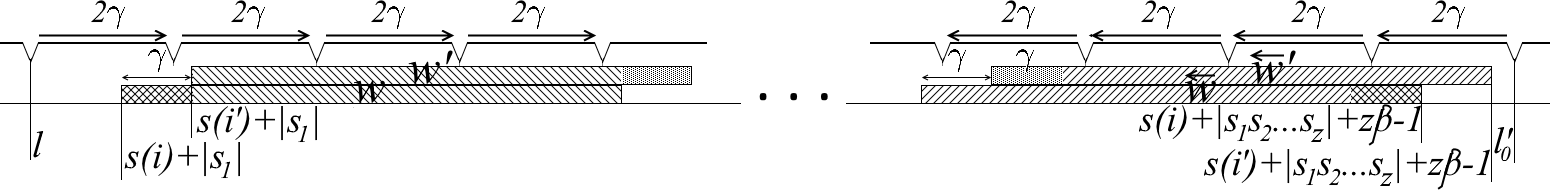}
\caption{\small Substitutions of $x$ and $\lvec{x}$ in two instances of $p$.}
\label{fig:respect}
\end{figure}

Then, for any $z' \in [1..r)$, the string $t[s(y\lceil\frac{\beta}8\rceil) + |s_1s_2\cdots s_{z'}| + (z' - 1)\beta .. s(y\lceil\frac{\beta}8\rceil) + |s_1s_2\cdots s_{z'}| + (z' - 1)\beta + \lceil\frac{\beta}2\rceil]$ is a substring of the substitutions of $x_{z'}$ in the instances of $p$ starting at $s(i)$ and $s(i')$. Therefore, this string has period $2\gamma \le \frac{\beta}4$. Applying Lemma~\ref{SubstringRun} for this string, we find in $O(1)$ time a run $t[i_{z'}..j_{z'}]$ with the minimal period $d_{z'} \le 2\gamma$ that contains the substitutions of $x_{z'}$ occurring at positions $s(i) + |s_1s_2\cdots s_{z'}| + (z' - 1)\beta$ and $s(i') + |s_1s_2\cdots s_{z'}| + (z' - 1)\beta$, respectively (note that $d_{z'}$ might not equal $2\gamma$); the run $t[i_{z'}..j_{z'}]$ must exist because otherwise there cannot exist such positions $i$ and $i'$ corresponding to two instances of $p$. The found runs $t[i_1..j_1], t[i_2..j_2], \ldots, t[i_{r-1}..j_{r-1}]$ are uniquely determined by $y \in (0..8]$ and their choice does not depend on the choice of $i$ or $i'$. Moreover, since the choice of $i$ and $i'$ was arbitrary, it follows that if, for some $i'' \in [(y - 1)\lceil\frac{\beta}{8}\rceil..y\lceil\frac{\beta}{8}\rceil)$, there is an instance of $p$ starting at position $s(i'')$ with substitutions of length $\beta$, then, for any $z' \in [1..r)$, the substitution of $x_{z'}$ in this instance is a substring of $t[i_{z'}..j_{z'}]$.

We check in $O(1)$ time whether $d_1 = d_2 = \cdots = d_{r-1}$ (if not, then there cannot exist such positions $i$ and $i'$ corresponding to two instances of $p$). Denote by $\ell$ and $\ell'_0$, respectively, the starting position of a Lyndon root of $t[i_1..j_1]$ and the ending position of a reversed Lyndon root of $t[i_z..j_z]$; $\ell$ and $\ell'_0$ can be computed in $O(1)$ time by Lemma~\ref{LyndonRoots}. Obviously, we necessarily have $\lrange{t[\ell..\ell{+}d{-}1]} = t[\ell'_0{-}d{+}1..\ell'_0]$. We check this condition using the $\lcp$ structure on the string $t\lvec{t}$. It follows from Lemma~\ref{PrimitiveVV} that the distance between $\ell$ and the starting position of $w$ in $t[i_1..j_1]$ must be equal to the distance between $\ell'_0$ and the ending position of $\lvec{w}$ in $t[i_z..j_z]$ modulo $d$, i.e., $s(i) + |s_1| - \ell \equiv \ell'_0 - (s(i) + |s_1s_2\cdots s_z| + z\beta - 1) \pmod{d}$ (see Fig.~\ref{fig:respect}). The latter is equivalent to the equality $2s(i) \equiv \ell + \ell'_0 - |s_1| - |s_1s_2\cdots s_z| - z\beta + 1 \pmod{d}$. The right hand side of this equality, denoted $\eta$, can be calculated in $O(r)$ time. Thus, since the choice of $i$ and $i'$ was arbitrary, we have $2s(i'') \equiv \eta \pmod{d}$ for any $i'' \in [(y - 1)\lceil\frac{\beta}{8}\rceil..y\lceil\frac{\beta}{8}\rceil)$ such that the string $t[s(i'') .. s(i'')+ |s_1s_2\cdots s_r| + (r - 1)\beta - 1]$ is an instance of $p$.

It turns out that, in a sense, the converse is also true. Suppose that $i \in [(y - 1)\lceil\frac{\beta}{8}\rceil..y\lceil\frac{\beta}{8}\rceil)$, $D[i] = 1$, $2s(i) \equiv \eta \pmod{d}$, and, for each $z' \in [1..r)$, the string $w_{z'} = t[s(i) + |s_1s_2\cdots s_{z'}| + (z' - 1)\beta .. s(i) + |s_1s_2\cdots s_{z'}| + z'\beta - 1]$, which is a ``suspected'' substitution of $x_{z'}$ in the string $t[s(i)..s(i) + |s_1s_2\cdots s_r| + (r - 1)\beta - 1]$, is contained in the run $t[i_{z'}..j_{z'}]$. Choose $z', z'' \in [1..r)$ such that $x_{z'} = x$ and $x_{z''} = \lvec{x}$. If $z' = 1$ and $z'' = z$, then the equalities $2s(i) \equiv \eta \pmod{d}$ and $\lrange{t[\ell..\ell{+}d{-}1]} = t[\ell'_0{-}d{+}1..\ell'_0]$ imply that $w_{z'} = w_{z''}$. It follows from the equality $D[i] = 1$ that, for any $z', z'' \in [1..r)$, if $x_{z'} = x_{z''}$, then $w_{z'} = w_{z''}$. Therefore, the string $t[s(i)..s(i) + |s_1s_2\cdots s_r| + (r - 1)\beta - 1]$ is an instance of $p$.

It is easy to verify that $2s(i) \equiv \eta \pmod{d}$ iff either $s(i) \equiv \frac{\eta}2 \pmod{d}$ or $s(i) \equiv \frac{\eta + d}2 \pmod{d}$. So, using an appropriate bit mask and the bitwise ``and'' operation on $\Theta(\log n)$-bit machine words, we can assign $D[i''] = 0$ for all $i'' \in [(y - 1)\lceil\frac{\beta}{8}\rceil..y\lceil\frac{\beta}{8}\rceil)$ such that $2s(i'') \not\equiv \eta \pmod{d}$ in $O(\frac{\beta}{\log n})$ time. Then, according to the starting and ending positions of the runs $t[i_1..j_1], \ldots, t[i_{r-1}..j_{r-1}]$, we calculate in an obvious way in $O(r)$ time the exact subrange $[k_1..k_2] \subset [(y - 1)\lceil\frac{\beta}{8}\rceil..y\lceil\frac{\beta}{8}\rceil)$ such that, for any $i \in [k_1..k_2]$ and any $z' \in [1..r)$, the string $t[s(i) + |s_1s_2\cdots s_{z'}| + (z' - 1)\beta .. s(i) + |s_1s_2\cdots s_{z'}| + z'\beta - 1]$ (a ``suspected'' substitution of $x_{z'}$) is a substring of $t[i_{z'}..j_{z'}]$. Finally, we fill the subarrays $D[(y - 1)\lceil\frac{\beta}{8}\rceil .. k_1 - 1]$ and $D[k_2 + 1 .. y\lceil\frac{\beta}{8}\rceil - 1]$ with zeros in $O(\frac{\beta}{\log n})$ time.

If something was wrong in the above scenario or simply the final array $D[(y - 1)\lceil\frac{\beta}{8}\rceil .. y\lceil\frac{\beta}{8}\rceil{-}1]$ contains only zeros, then we still can have exactly one position $i \in [(y - 1)\lceil\frac{\beta}{8}\rceil .. y\lceil\frac{\beta}{8}\rceil)$ such that $t[s(i)..s(i) + |s_1s_2\cdots s_r| + (r - 1)\beta - 1]$ is an instance of $p$. The minimal period of the substitution of $x$ in such instance will necessarily be greater than $\frac{\beta}4$ (otherwise, we would find this instance by the above filtering algorithm). Note that such instance of $p$ should contain the substring $t[s(y\lceil\frac{\beta}{8}\rceil) + |s_1| .. s(y\lceil\frac{\beta}{8}\rceil) + |s_1| + \lceil\frac{\beta}{2}\rceil]$ in the substitution of $x_1$. Thus, the substitution of $x_{z-1} = x$ (recall that $x_1 = x$ and $z$ is the minimal number such that $x_z = \lvec{x}$) in such instance should contain the substring $\mu = t[q .. q + \lceil\frac{\beta}{2}\rceil]$, where $q = s(y\lceil\frac{\beta}{8}\rceil) + |s_1s_2\cdots s_{z-1}| + (z - 2)\beta$.

Suppose that the minimal period of $\mu$ is greater than $\frac{\beta}4$; since $\frac{\beta}4 \le \frac{|\mu|}2$, this condition can be checked in $O(1)$ time by Lemma~\ref{SubstringRun}. Since $\mu$ is supposed to be a substring of a substitution of $x_{z-1} = x$ in the required instance of $p$, the string $\lvec{\mu}$ must occur in the string $t[q + \lceil\frac{\beta}{2}\rceil + |s_z| .. q + 2\beta + |s_z|]$ as a substring of a substitution of $x_z = \lvec{x}$ (of course, if there exists such instance of $p$). Since the minimal period of $\lvec{\mu}$ is greater than $\frac{\beta}4$, any two occurrence of $\lvec{\mu}$ cannot overlap on $|\mu| - \frac{\beta}4 \ge \frac{\beta}4$ letters. Therefore, there are at most $2\beta / \frac{\beta}4 = 8$ occurrence of $\lvec{\mu}$ in $t[q + \lceil\frac{\beta}{2}\rceil + |s_z| .. q + 2\beta + |s_z|]$. We find all these occurrences of $\lvec{\mu}$ using a slightly modified algorithm from the proof of Lemma~\ref{vFind} putting $\lambda = |s_z|$; it takes $O(r + \frac{r|\mu|}{\log n})$ time plus $O(\frac{\log n}{\log\log n})$ time, if $\frac{\log}{16\log\log n} \le |v| \le \log n$. Each found occurrence of $\lvec{\mu}$ specifies exactly one possible instance of $p$ with substitutions of length $\beta$ and the starting position of this instance can be easily calculated. To test whether the substitutions of the variables $x$ and $\lvec{x}$ in this instance respect each other, we utilize the $\lcp$ structure.

Finally, suppose that the minimal period of $\mu$ is less than or equal to $\frac{\beta}4$. Since $\frac{\beta}4 \le \frac{|\mu|}2$, by Lemma~\ref{SubstringRun}, we can find in $O(1)$ time a run $t[i''..j'']$ containing $\mu$ and having the same minimal period as $\mu$. Denote $\mu_1 = t[i'' - 1 .. s(y\lceil\frac{\beta}{8}\rceil) + |s_1s_2\cdots s_{z-1}| + (z - 2)\beta + y\lceil\frac{\beta}{2}\rceil]$ and $\mu_2  = t[s(y\lceil\frac{\beta}{8}\rceil) + |s_1s_2\cdots s_{z-1}| + (z - 2)\beta .. j'' + 1]$. Note that $\mu_1$ (resp.,~$\mu_2$) is the minimal extension of $\mu$ to the left (resp.,~right) that ``breaks'' the minimal period of $\mu$. One can show that both $\mu_1$ and $\mu_2$ are not periodic. Since the minimal period of the substitution of $x$ in the instance of $p$ that we are searching for must be greater than $\frac{\beta}4$, this substitution must contain either $\mu_1$ or $\mu_2$. So, to find this instance, it suffices to execute the algorithm similar to that described above putting $\mu = \mu_1$ and $\mu = \mu_2$.\qed
\end{proof}

\subsection*{To Section~\ref{SectPeriodic}}

\paragraph{Three and more runs.}
Here in appendix we also consider the case when $p$ contains both $x$ and $\lvec{x}$. As above, this case turned out to be quite difficult and we need additional machinery for this. The following result is Lemma~14 from~\cite{KosolobovPalk}.

\begin{lemma}
For any primitive string $w$, there exists at most one pair of palindromes $u$ and $v$ such that $v\ne\epsilon$ and $w = uv$.\label{PrimitivePalPair}
\end{lemma}

\begin{lemma}
Suppose that two strings $w$ and $\lvec{w}$ of length ${\ge}3d$ both lie in a run with the minimal period $d$. Then, there exists a unique pair of palindromes $u,v$ such that $|uv| = d$, $v \ne \epsilon$, and $\lvec{w}$ is a prefix of the infinite strings $(vu)^{\infty}$.\label{Palsuv}
\end{lemma}
\begin{proof}
Since $\lvec{w}[1..d]$ and $w[1..d]$ are substrings of the same run with period $d$, we have $\lvec{w}[1..d] = vu$ and $w[1..d] = uv$ for some strings $u$ and $v$ such that $v \ne \epsilon$, i.e., $vu = \lrange{uv} = \lvec{v}\lvec{u}$. Hence, $u$ and $v$ are palindromes. By Lemma~\ref{PrimitivePalPair}, this pair of palindromes is unique.\qed
\end{proof}

\begin{lemma}
Let $t[i'..j']$ be a run with period $d$. If, for $h \in [i'{+}d..j'{+}1]$, $t[h{-}d..h{-}1] = uv$ for some palindromes $u$ and $v$ such that $v\ne\epsilon$, then, for any $h' \in [i'{+}d..j'{+}1]$, we have $t[h'{-}d..h'{-}1] = u'v'$ for palindromes $u'$ and $v'$ such that $|u'| = (|u| - 2(h' - h)) \bmod d$.\label{InARunPalPair}
\end{lemma}
\begin{proof}
Let $h' = h + 1$ (the case $h' = h - 1$ is analogous). It suffices to prove that $t[h'{-}d..h'{-}1] = u'v'$ for palindromes $u'$ and $v'$ such that $|u'| = (|u| - 2) \bmod d$. Since $t[h'{-}d..h'{-}1]$ is a suffix of $uv\cdot u[1]$ (or $v\cdot v[1]$ if $|u| = 0$), we have $t[h'{-}d..h'{-}1] = u'v'$, where $u' = u, v' = v$, if $d = 1$, and $u', v'$ are defined as follows, if $d > 1$:\\
\indent $u = u[1]u'u[1]$ if $|u| > 1$,  $u' = v$ if $|u| = 1$,  $v = v[1]u'v[1]$ if $|u| = 0$,\\
\indent $v' = u[1]vu[1]$ if $|u| > 1$,  $v' = u$ if $|u| = 1$,  $v' = v[1]v[1]$ if $|u| = 0$.\qed
\end{proof}

\newcommand{\len}{\mathop{\mathsf{len}}}
\newcommand{\pal}{\mathop{\mathsf{pal}}}
\newcommand{\link}{\mathop{\mathsf{link}}}
\newcommand{\seriesLink}{\mathop{\mathsf{seriesLink}}}

\begin{lemma}
Assuming $O(n)$ time preprocessing, one can find for any substring $t[i..j]$ in $O(1)$ time a pair of palindromes $u, v$ such that $t[i..j] = uv$ and $v\ne \epsilon$ or decide that there is no such pair.\label{FindPalPair}
\end{lemma}
\begin{proof}
In Lemma~C4 from~\cite{GalilSeiferas2},
if there exist palindromes $u$ and $v$ such that $t[i..j] = uv$, then there exist palindromes $u'$ and $v'$ such that $t[i..j] = u'v'$ and either $u'$ is the longest palindromic prefix of $t[i..j]$ or $v'$ is the longest palindromic suffix of $t[i..j]$. To test in $O(1)$ time whether a given substring is a palindrome, we can use the data preprocessed by Manacher's algorithm~\cite{Manacher};
so, it suffices to describe a data structure that allows to find the longest palindromic prefix/suffix of any substring in $O(1)$ time. Without loss of generality, we consider the case of palindromic suffixes.

Our main tool is the data structure called \emph{eertree}, which was introduced in~\cite{RubinchikShur}.
The eertree of $t$ can be built in $O(n)$ time by~Proposition~11 in the paper of Rubinchik and Shur. The main body of eertree of $t$ consists of nodes; any node $a$ represents a palindrome $\pal[a]$ that is a substring of $t$ and, conversely, any palindrome that is a substring of $t$ is represented by some node. Denote by $\link[a]$ the node representing the longest proper palindromic suffix of $\pal[a]$ (if any). In Proposition~9 in the paper of Rubinchik and Shur, for each node $a$, there was defined a \emph{series link} $\seriesLink[a]$ such that $\seriesLink[a]$ is either a node $a'$ representing the longest palindromic suffix of $\pal[a]$ such that $|\pal[a]| - |\pal[\link[a]]| \ne |\pal[a']| - |\pal[\link[a']]|$ or the node representing the empty palindrome if there is no such $a'$.

In~\cite{RubinchikShur}
it was shown that the tree that is induced by the series links with the root in the node representing the empty palindrome has height at most $O(\log n)$. We build on this tree the \emph{weighted ancestor} data structure from Lemma~11 of~\cite{GawrychowskiLewensteinNicholson}
that allows, for any given node $a$ and number $\gamma \ge 0$, to find the farthest ancestor $a'$ of $a$ such that $|\pal[a']| \ge \gamma$; moreover, since the height of our tree is $O(\log n)$, we can answer the queries in $O(1)$ time. Finally, as it was proved in~\cite{RubinchikShur},
during the construction of eertree, we can create an array $\mathsf{psuf}[1..n]$ such that, for any $j \in [1..n]$, $\mathsf{psuf}[j]$ is the node of eertree representing the longest palindromic suffix of $t[1..j]$.

Now, to find the longest palindromic suffix of a given substring $t[i..j]$, we compute the farthest ancestor $a$ of $\mathsf{psuf}[j]$ such that $|\pal[a]| \ge j - i + 1$; then, by the definition of the series links, the longest palindromic suffix of $t[i..j]$ is either $\pal[\seriesLink[a]]$ or the palindromic suffix of $\pal[a]$ with the length $|\pal[a]| - c(|\pal[a]| - |\pal[\link[a]]|)$, where $c$ is the minimal integer such that $|\pal[a]| - c(|\pal[a]| - |\pal[\link[a]]|) \le j - i + 1$.\qed
\end{proof}

Now we can prove Lemma~\ref{SeparateLemma} including the case when the pattern $p$ contains both $x$ and $\lvec{x}$.

\begin{proof}[Lemma~\ref{SeparateLemma}]
We can find a run $t[i'..j']$ with period $d$ containing $v = t[h_1..h_2]$ in $O(1)$ time using Lemma~\ref{SubstringRun}. Let $t[i..j] = s_1w_1s_2w_2\cdots s_{r-1}w_{r-1}s_r$ be an instance of $p$ satisfying the conditions in the statement of the lemma. Let us find all possible runs with period $d$ that can contain $w_z$. Denote by $h$ the starting position of $s_z$. Since $|w_1| \ge \frac{3}2|v| \ge 3d$, Lemma~\ref{wsw2} implies that $h = j' - |\pref_d(s_z)| + 1$ or $h \in (j'{+}1{-}d..j'{+}1]$ or $h \in (j'{-}|s_z|{-}d..j'{-}|s_z|]$. Suppose that $h \in (j'{+}1{-}d..j'{+}1]$ (resp., $h \in (j'{-}|s_z|{-}d..j'{-}|s_z|]$, $h = j' - |\suff_d(s_z)| + 1$). Then, any run containing $w_z$ certainly contains the substring $t[j' + |s_z| + 1 .. j' + |s_z| + 2d]$ (resp., $t[j'..j' + 2d - 1]$, $t[h + |s_z| .. h + |s_z| + 2d]$). We find a run with period $d$ containing this substring in $O(1)$ time by Lemma~\ref{SubstringRun}. Thus, we have three possible locations for a run with period $d$ containing $w_z$ (note that some of the found runs can coincide). We process each separately; let $t[i''..j'']$ be one of these three runs. So, suppose that $t[i'..j']$ and $t[i''..j'']$ are runs with period $d$ containing the substrings $w_1s_2w_2\cdots s_{z-1}w_{z-1}$ and $w_zs_{z+1}w_{z+1}\cdots s_{z'-1}w_{z'-1}$, respectively.

Denote by $h'$ the starting position of $s_{z'}$. Note that $t[h..h'{-}1] = s_zw_z\cdots s_{z'-1}w_{z'-1}$. Hence, if the number $h' - h$ is known, we can calculate $|w_1| = \frac{h' - h - |s_zs_{z+1}\cdots s_{z'-1}|}{z' - z}$, apply Lemma~\ref{FixedLen} for $\beta = |w_1|$, and thus find all corresponding instances of $p$ in $O(r + \frac{r|v|}{\log n})$ time plus $O(\frac{\log n}{\log\log n})$ time, if $\frac{\log}{16\log\log n}{\le}|v|{\le}\log n$. So, our aim is to find a constant number of possible values for $h' - h$ and process each of them with the help of Lemma~\ref{FixedLen}.

\begin{figure}[htb]
\centering
    \vskip-4mm
    \includegraphics[scale=0.38]{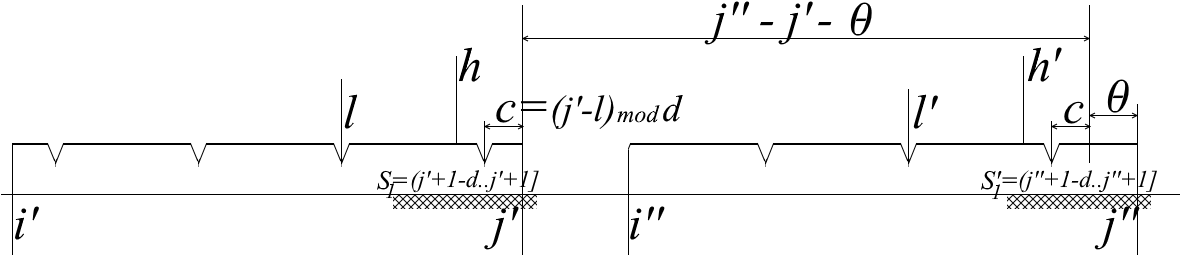}
    \vskip-3mm
    \caption{\small The case $h' - h = j'' - j' - \theta$ in the proof of Lemma~\ref{SeparateLemma}.}
    \label{fig:separate}
    \vskip-6mm
\end{figure}

Suppose that $x_{z-1} = x_{z'-1}$. By Lemma~\ref{PrimitiveVV}, since $t[h{-}d..h{-}1] = t[h'{-}d..h'{-}1]$, we have $h - \ell \equiv h' - \ell' \pmod{d}$, where $\ell$ and $\ell'$ are the starting positions of Lyndon roots of, resp., $t[i'..j']$ and $t[i''..j'']$; $\ell$ and $\ell'$ can be found in $O(1)$ time by Lemma~\ref{LyndonRoots}. So, we obtain $h' - h \equiv \ell' - \ell \pmod{d}$. By Lemma~\ref{wsw2}, $h$ either equals $j' - |\pref_d(s_z)| + 1$ or lies in one of the segments $S_1 = (j'{+}1{-}d..j'{+}1]$ or $S_2 = (j'{-}|s_z|{-}d..j'{-}|s_z|]$ of length $d$; similarly, $h'$ either equals $j'' - |\pref_d(s_{z'})| + 1$ or lies in $S'_1 = (j''{+}1{-}d..j''{+}1]$ or $S'_2 = (j''{-}|s_{z'}|{-}d..j''{-}|s_{z'}|]$.
For each $h \in S_1$, there is exactly one $h' \in S'_1$ such that $h' - h \equiv \ell' - \ell \pmod{d}$; moreover, one can easily prove that in this case $h' - h$ is equal to either $j'' - j' - \theta$ or $j'' - j' - \theta + d$, where $\theta = ((j'' - \ell') - (j' - \ell)) \bmod d$ (see Fig.~\ref{fig:separate}). For each of these values of $h' - h$, we apply Lemma~\ref{FixedLen} putting $\beta = \frac{h' - h - |s_zs_{z+1}\cdots s_{z'-1}|}{z' - z}$. Other combinations ($h \in S_1$ and $h' \in S'_2$, $h \in S_2$ and $h' \in S'_1$, $h \in S_2$ and $h' \in S'_2$) are analogous; the cases when either $h = j' - |\pref_d(s_z)| + 1$ or $h' = j'' - |\pref_d(s_{z'})| + 1$ are even simpler. It remains to discuss the case $x_{z-1} \ne x_{z'-1}$.

Assume that $x_{z-1} \ne x_{z'-1}$. Suppose that $x_{z} = x_{z'}$. Let us find all runs with period $d$ that can contain the substitution $w_{z'}$. Using the run $t[i''..j'']$, one can find three possible choices for such run in the same way as we found three choices for $t[i''..j'']$ using $t[i'..j']$; we should process each of these three choices. Let us fix one such run $t[i'''..j''']$. So, suppose that $t[i'''..j''']$ contains $w_{z'}$. Denote $h_0 = h + |s_z| - 1$ and $h'_0 = h' + |s_{z'}| - 1$. Since $t[h_0{+}1..h'_0] = w_zs_{z+1}w_{z+1}\cdots w_{z'-1}s_{z'}$, if the number $h'_0 - h_0$ is known, we can calculate $|w_1| = (h'_0 - h_0 - |s_{z+1}s_{z+2}\cdots s_{z'}|) / (z' - z)$ and apply Lemma~\ref{FixedLen} putting $\beta = |w_1|$. Let $\ell''$ be the starting position of a Lyndon root of $t[i'''..j''']$; $\ell''$ can be found in $O(1)$ time by Lemma~\ref{LyndonRoots}.  Now we can find a constant number of possible values for the number $h'_0 - h_0$ doing the same case analysis on the positions $h_0$ and $h'_0$ as we did on $h$ and $h'$ but using the runs $t[i''..j'']$ and $t[i'''..j''']$ instead of $t[i'..j']$ and $t[i''..j'']$, the positions $\ell'$ and $\ell''$ instead of $\ell$ and $\ell'$, and a reversed version of Lemma~\ref{wsw2}.

Finally, suppose that $x_{z-1} \ne x_{z'-1}$ and $x_z \ne x_{z'}$. Without loss of generality, assume that $x_{z-1} = x$. Then, we have either $x_{z-1}s_zx_z = xs_z\lvec{x}$ and $x_{z'-1}s_{z'}x_{z'} = \lvec{x}s_{z'}x$ or $x_{z-1}s_zx_z = xs_zx$ and $x_{z'-1}s_{z'}x_{z'} = \lvec{x}s_{z'}\lvec{x}$.

Suppose that $x_{z-1}s_zx_z = xs_z\lvec{x}$ and $x_{z'-1}s_{z'}x_{z'} = \lvec{x}s_{z'}x$. Since $t[h'{-}d..h'{-}1] = \lrange{t[h'{+}|s_{z'}|..h'{+}|s_{z'}|{+}d{-}1]}$, it follows from Lemma~\ref{PrimitiveVV} that $h' - 1 - \ell' \equiv \ell''_0 - (h' + |s_{z'}|) \pmod{d}$, where $\ell''_0$ is the ending position of a reversed Lyndon root of $t[i'''..j''']$; $\ell''_0$ can be computed in $O(1)$ time by Lemma~\ref{LyndonRoots}. Thus, we obtain $2h' \equiv \ell' + \ell''_0 + 1 - |s_{z'}| \pmod{d}$, i.e., either $h' \equiv \frac{\ell' +\ell''_0 + 1 - |s_{z'}|}2 \pmod{d}$ or $h' \equiv \frac{\ell' + \ell''_0 + 1 - |s_{z'}| + d}2 \pmod{d}$. It is easy to see that any segment of length $d$ contains at most two positions $h'$ satisfying the latter equalities; moreover, one can find these positions in $O(1)$ time. So, since, by Lemma~\ref{wsw2}, $h' = j' - |\pref_d(s_z)| + 1$ or $h'$ lies in one of two segments of length $d$, there are at most five values for $h'$ such that $2h' \equiv \ell' + \ell''_0 + 1 - |s_{z'}| \pmod{d}$; we can find them all in $O(1)$ time. Symmetrically, we find at most five possible values for $h$ but using the runs $t[i'..j']$ and $t[i''..j'']$ instead of $t[i''..j'']$ and $t[i'''..j''']$ and the position $\ell'_0$ instead of $\ell''_0$. Finally, for each found value $h' - h$, we calculate $\beta = (h' - h - |s_zs_{z+1}\cdots s_{z'-1}|) / (z' - z)$ and apply Lemma~\ref{FixedLen}.

It remains to process the case when $x_{z-1}s_zx_z = xs_zx$ and $x_{z'-1}s_{z'}x_{z'} = \lvec{x}s_{z'}\lvec{x}$. Since in this case $w$ and $\lvec{w}$ both are substrings of $t[i''..j'']$ as substitutions of $x_z = x$ and $x_{z'-1} = \lvec{x}$, by Lemma~\ref{Palsuv}, there exist palindromes $u$ and $v$ such that $|uv| = d$, $v \ne \epsilon$, and $\lvec{w}$ is a prefix of the infinite string $(vu)^{\infty}$. In $O(r)$ time we find a number $z'' \in (z..z')$ such that $x_{z''-1}s_{z''}x_{z''} = xs_{z''}\lvec{x}$. Since $ws_{z''}\lvec{w}$ is a substring of $t[i''..j'']$, we have $s_{z''} = u(vu)^{k'}$ for some $k' \ge 0$. Therefore, we can compute the length of $u$: $|u| = |s_{z''}| \bmod d$. Given a segment $[f_1..f_2] \subset [i'{+}d..j'{+}1]$ of length $d$, it follows from Lemmas~\ref{PrimitivePalPair} and~\ref{InARunPalPair} that there exist at most two positions $f \in [f_1..f_2]$ such that $t[f{-}d..f{-}1] = v'u'$ for some palindromes $u'$ and $v'$ such that $|u'| = |u|$. We can find these positions in $O(1)$ time using the equality from Lemma~\ref{InARunPalPair} provided, for some $f' \in [f_1..f_2]$, we know palindromes $u''$ and $v''$ such that $t[f'{-}d..f'{-}1] = u''v''$. The palindromes $u''$ and $v''$ can be computed for arbitrary $f'$ in $O(1)$ time by Lemma~\ref{FindPalPair}. Since, by Lemma~\ref{wsw2}, $h$ either equals $j' - |\pref_d(s_z)| + 1$ or lies in one of the segments $(j'{+}1{-}d..j'{+}1]$ or $(j'{-}|s_z|{-}d..j'{-}|s_z|]$ of length $d$, there are at most five possible values for $h$ and each of them can be found in $O(1)$ time. In the same way we find at most five possible values for $h'$. Finally, for each obtained possible value $h' - h$, we calculate $\beta = (h' - h - |s_zs_{z+1}\cdots s_{z'-1}|) / (z' - z)$ and apply Lemma~\ref{FixedLen}.\qed
\end{proof}

\paragraph{In-a-run instances of $p$.}

The special case $p = s_1\lvec{x}s_2xs_3$ is discussed below.

To prove Lemma~\ref{GeneralInRunLemma}, we need the following auxiliary result.
\begin{lemma}
Suppose that $x_1 = \cdots = x_{r-1}$. Given a run $t[i'..j']$ with period $d$, a segment $[b_1..b_2] \subset [i'..j']$ of length $d$, and an integer $\eta \ge d$, we can compute in $O(r + \frac{rd}{\log n})$ time a bit array $D'[b_1..b_2]$ such that, for $h \in [b_1..b_2]$, $D'[h] = 1$ iff the string $t[h{-}|s_1s_2\cdots s_{r-1}|{-}(r{-}1)\eta..h{+}|s_r|{-}1]$ is an instance of $p$ and $h - |s_2s_3\cdots s_{r-1}| - (r - 1)\eta \ge i'$.\label{Dsubarray}
\end{lemma}
\begin{proof}
We obtain a bit array $D'[b_1..b_2]$ performing in $O(r{+}\frac{rd}{\log n})$ time the bitwise ``and'' of the subarrays $\{D_{z}[b_1{-}\gamma_z .. b_2{-}\gamma_z]\}_{z=1}^{r}$, where $\gamma_z = (r{-}z)\eta{-}|s_{z}s_{z+1}\cdots s_{r-1}|$. In $O(\frac{d}{\log n})$ time we fill with zeros a subarray $D'[b_1..b]$ for the maximal $b \in [b_1..b_2]$ such that $b{-}(r{-}1)\eta{-}|s_2s_3\cdots s_{r-1}| < i'$ (if any). Since $x_1{=}\cdots{=}x_{r-1}$ and $\eta \ge d$, it follows from Lemma~\ref{PrimitiveVV} that, for $h \in [b_1..b_2]$ such that $D'[h] = 1$, $t[h{-}|s_1s_2\cdots s_{r-1}|{-}(r{-}1)\eta..h{+}|s_r|{-}1]$ is an instance of $p$ iff $\eta \equiv -|s_z| \pmod{d}$ for all $z \in (1..r)$ (i.e., if substitutions of $x$ are aligned properly). So, if $\eta \not\equiv -|s_z| \pmod{d}$ for some $z \in (1..r)$, then we fill $D'[b_1..b_2]$ with zeros.\qed
\end{proof}

\begin{proof}[Lemma~\ref{GeneralInRunLemma}]
Choose $h \in [b_1..b_2]$. Denote $i_h = h - |s_1s_2\cdots s_{r-1}| - (r - 1)\delta$, $j_h = h + |s_r| - 1$, and $c_h = \lfloor\frac{i_h - i'}{(r - 1)d}\rfloor$. By Lemma~\ref{AllInRunOccurrences}, if a string $t[i..j] = s_1w_1s_2w_2\cdots s_{r-1}w_{r-1}s_r$ satisfies~\eqref{eq:mainprop}, $i' \le i \le j - |s_r| + 1 = h$, and $|w_1| \equiv \delta \pmod{d}$, then $t[i_h..j_h]$ is an instance of $p$ and $t[i..j] = t[i_h{-}(r{-}1)cd..j_h]$ for some $c \in [0..c_h]$; conversely, if $t[i_h..j_h]$ is an instance of $p$, then all substrings $t[i_h{-}(r{-}1)cd..j_h]$, for $c \in [0..c_h]$, are instances of $p$. By the definition of $c_h$, there is a threshold $h' \in [b_1..b_2]$ such that, for any $h_1, h_2 \in [b_1..b_2]$, $c_{h_1} = c_{h_2}$ if $h_1, h_2 \in [b_1..h']$ or $h_1, h_2 \in (h'..b_2]$, and $|c_{h_1} - c_{h_2}| = 1$, otherwise; $h'$ can be found in $O(1)$ time by simple calculations. Put $d' = c_{h'}$ and $d'' = c_{h'+1}$. Suppose that $x_1{=}\cdots{=}x_{r-1}$. Then, applying Lemma~\ref{Dsubarray} for $\eta = \delta$, we compute the required bit array $E[b_1..b_2]$. Now it remains to find all strings $t[i..j]$ satisfying~\eqref{eq:mainprop} and such that $i < i'$ and $j - |s_r| + 1 \in [b_1..b_2]$.

Let $t[i..j] = s_1w_1\cdots s_{r-1}w_{r-1}s_r$ satisfy~\eqref{eq:mainprop}, $j - |s_r| + 1 \in [b_1..b_2]$, $|w_1| \equiv \delta \pmod{d}$, and $i < i'$. Denote $h = j - |s_r| + 1$. By a symmetric version of Lemma~\ref{wsw}, we have either $i + |s_1| \in [i'..i'{+}d)$ or $i + |s_1| = i' + |\suff_d(s_1)|$. Suppose that $i + |s_1| \in [i'..i'{+}d)$. By Lemma~\ref{AllInRunOccurrences}, we have $i + |s_1| = h - |s_2s_3\cdots s_{r-1}| - (r-1)(\delta + a_hd)$, where $a_h$ is the maximal integer such that $h - |s_2s_3\cdots s_{r-1}| - (r-1)(\delta + a_hd) \ge i'$. The definition of $a_h$ implies that there is $h'' \in [b_1..b_2]$ such that, for any $h_1,h_2 \in [b_1..b_2]$, $a_{h_1} = a_{h_2}$ if $h_1,h_2\in[b_1..h'']$ or $h_1,h_2\in(h''..b_2]$, and $|a_{h_1} - a_{h_2}| = 1$, otherwise; $h''$ can be simply found in $O(1)$ time. Put $a' = a_{h''}$ and $a'' = a_{h''+1}$. Suppose that $x_1{=}\cdots{=}x_{r-1}$. Then, we compute two bit arrays $F_1[b_1..b_2]$ and $F_2[b_1..b_2]$ applying Lemma~\ref{Dsubarray} for $\eta = \delta + a'd$ and $\eta = \delta + a''d$, respectively. Finally, we concatenate the arrays $F_1[b_1..h'']$ and $F_2[h''{+}1..b_2]$ to obtain $F[b_1..b_2]$.

Suppose that $i + |s_1| = i' + |\suff_d(s_1)|$. Since $|t[i{+}|s_1|..h{-}1]| = |w_1s_2w_2\cdots s_{r-1}w_{r-1}|$, we have $h - (i' + |\suff_d(s_1)|) \equiv |s_2s_3\cdots s_{r-1}| + (r - 1)\delta \pmod{d}$. Since $b_2 - b_1 + 1 = d$, there is exactly one position $h \in [b_1..b_2]$ satisfying the latter equality; $h$ can be found in $O(1)$ time. We check whether $t[i' + |\suff_d(s_1)| - |s_1| .. h + |s_r| - 1]$ is an instance of $p$ in $O(r)$ time using the $\lcp$ structure and the arrays $\{D_z\}_{z=1}^r$; thus, we may find an additional instance of $p$ that is not encoded in $E$ and $F$.

While in the case $x_1 = \cdots = x_{r-1}$ it was sufficient to rely on the periodicity of $t[i'..j']$ to test whether corresponding substitutions are equal (as in Lemma~\ref{Dsubarray}), in the case when $p$ contains both $x$ and $\lvec{x}$ it is not clear how to test for all $h \in [b_1..b_2]$ simultaneously whether corresponding substitutions of $x$ and $\lvec{x}$ respect each other. However, it turns out that there are at most two positions $h \in [b_1..b_2]$ for which there might exist a string $t[i..j]$ satisfying~\eqref{eq:mainprop} and such that $j - |s_r| + 1 = h$. We find these two positions in $O(1)$ time and process each of them separately in $O(r)$ time.

Let $t[i..j]$ be a string satisfying~\eqref{eq:mainprop} and such that $j - |s_r| + 1 \in [b_1..b_2]$. Denote by $w$ the substitution of $x$ in $t[i..j]$. By Lemma~\ref{Palsuv}, there exist palindromes $u$ and $v$ such that $v \ne \epsilon$ and $\lvec{w}[1..d] = vu$. Let us find the lengths of $u$ and $v$. Choose a number $z' \in (1..r)$ such that $x_{z'-1}s_{z'}x_{z'} = \lvec{x}s_{z'}x$ (it exists because $x_{r-1} = x$). Since $p \ne s_1xs_2\lvec{x}s_3$, $p \ne s_1\lvec{x}s_2xs_3$, and $r \ge 3$, there is $z'' \in (1..r)$ such that $x_{z''-1}s_{z''}x_{z''}$ is equal to either $xs_{z''}\lvec{x}$ or one of the strings $xs_{z''}x$ or $\lvec{x}s_{z''}\lvec{x}$. Suppose that $x_{z''-1}s_{z''}x_{z''} = xs_{z''}\lvec{x}$. Since the strings $vu$ and $uv$ are primitive, it follows from Lemma~\ref{PrimitiveVV} that $s_{z''} = u(vu)^{k'}$ for an integer $k'$. Therefore, we can compute the length of $u$: $|u| = |s_{z''}| \bmod d$. Now suppose that $x_{z''-1}s_{z''}x_{z''} = xs_{z''}x$ (resp. $x_{z''-1}s_{z''}x_{z''} = \lvec{x}s_{z''}\lvec{x}$). It follows from Lemma~\ref{PrimitiveVV} that the distance between any two occurrence of $w$ (resp., $\lvec{w}$) in $t[i'..j']$ is a multiple of $d$; thus, we have $|w| \equiv -|s_{z''}| \pmod{d}$. Since $\lvec{w}s_{z'}w$ is a substring of $t[i'..j']$, by Lemma~\ref{PrimitiveVV}, we have $\lvec{w}s_{z'}w = (vu)^{k'}v$ for an integer $k'$. Therefore, we can compute the length of $v$: $|v| = (|s_{z'}| + 2|w|) \bmod d = (|s_{z'}| - 2|s_{z''}|) \bmod d$ assuming $|v| = d$ if $|s_{z'}| - 2|s_{z''}| \equiv 0 \pmod{d}$.

Using Lemma~\ref{FindPalPair}, we find in $O(1)$ time palindromes $u''$ and $v''$ such that $t[b_2{-}d{+}1..b_2] = u''v''$ (they exist by Lemma~\ref{InARunPalPair}). Since $b_2 - b_1 + 1 = d$, it follows from Lemmas~\ref{PrimitivePalPair} and~\ref{InARunPalPair} that there are at most two positions $h \in [b_1..b_2]$ such that $t[h{-}d..h{-}1] = u'v'$ for some palindromes $u'$ and $v'$ satisfying $|u'| = |u|$ and $|v'| = |v|$; these positions $h$ can be found using the equality from Lemma~\ref{InARunPalPair} and the lengths $|u''|$ and $|v''|$. Fix one such position $h \in [b_1..b_2]$.

By Lemma~\ref{AllInRunOccurrences}, there exists an instance $t[i..j] = s_1w_1s_2w_2\cdots w_{r-1}s_{r-1}s_r$ of $p$ satisfying~\eqref{eq:mainprop} and such that $i \ge i'$, $|w_1| \equiv \delta \pmod{d}$, and $j - |s_r| + 1 = h$ iff the string $t[i_h .. j_h]$ is an instance of $p$ (recall that $i_h = h - |s_1s_2\cdots s_{r-1}| - (r - 1)\delta$, $j_h = h + |s_r| - 1$, and $c_h = \lfloor\frac{i_h - i'}{(r - 1)d}\rfloor$); moreover, in this case $t[i..j] = t[i_h - (r - 1)(\delta + cd) .. j_h]$ for some $c \in [0..c_h]$. So, we test whether $t[i_h .. j_h]$ is an instance of $p$ in $O(r)$ time using the arrays $\{D_z\}_{z=1}^r$ and the $\lcp$ structure of the string $t\lvec{t}$. So, if $t[i_h..j_h]$ is an instance of $p$, then we put $E[h] = 1$; if $h$ is the smallest of the two positions such that $t[h{-}d..h{-}1] = u'v'$ for some palindromes $u'$ and $v'$ such that $|u'| = |u|$ and $|v'| = |v|$, then we put $h' = h$ and $d' = c_h$; otherwise, we put $d'' = c_h$. (So, $E[b_1..b_2]$ contains at most two non-zero positions).

To find all instances $t[i..j]$ of $p$ such that $j - |s_r| + 1 = h$ and $i < i'$, we use a case analysis relying on a symmetric version of Lemma~\ref{wsw} similar to the analysis described above for the case $x_1 = \cdots = x_{r-1}$. By Lemma~\ref{wsw}, we have either $i + |s_1| = i' + |\suff_d(s_1)|$ or $i + |s_1| \in [i'..i'{+}d)$. First, suppose that $i + |s_1| \in [i'..i'{+}d)$. Since we must have $h - (i + |s_1|) \equiv |s_2s_3\cdots s_{r-1}| + (r-1)\delta \pmod{d}$, we find in $O(1)$ time at most one possible position $h' \in [i'..i'{+}d)$ such that $h - h' \equiv |s_2s_3\cdots s_{r-1}| + (r-1)\delta \pmod{d}$ (we suspect that $h' = i + |s_1|$) and test whether $t[h' - |s_1| .. h + |s_r| - 1]$ is an instance of $p$ in $O(r)$ time with the aid of the arrays $\{D_z\}_{z=1}^r$ and the $\lcp$ structure of $t\lvec{t}$; if this string is an instance, then we put $F[h] = 1$, $a' = (h - h' - |s_2s_3\cdots s_{r-1}| - (r - 1)\delta) / d$, $h'' = b_2$. (So, $F$ trivially encodes at most one instance of $p$) Finally, we test in $O(r)$ time whether $t[i' + |\suff_d(s_1)| - |s_1| .. h + |s_r| - 1]$ is an instance of $p$; thus, we can find an additional instance of $p$ that is not encoded in $E$ or $F$.\qed
\end{proof}

\paragraph{In-a-run instances of $p$: the special case $p = s_1\lvec{x}s_2xs_3$.}
Consider the special case $p = s_1\lvec{x}s_2xs_3$. Let us count all instances $t[i..j] = s_1\lvec{w}s_2ws_3$ of $p$ satisfying~\eqref{eq:mainprop} and such that $i + |s_1\lvec{w}| \in [b_1..b_2]$. Suppose that $t[i..j] = s_1\lvec{w}s_2ws_3$ is such instance. By Lemma~\ref{Palsuv}, there exist palindromes $u$ and $v$ such that $v \ne \epsilon$ and $w[1..d] = uv$. It follows from Lemma~\ref{PrimitiveVV} that $s_2 = (vu)^{k'}v$ for some integer $k' \ge 0$. Hence, we can calculate the length of $v$: $|v| = |s_2| \bmod d$ assuming $|v| = d$ if $|s_2| \bmod d = 0$. Using Lemma~\ref{FindPalPair}, we find in $O(1)$ time palindromes $u''$ and $v''$ such that $t[b_2{-}d{+}1..b_2] = u''v''$ (they exist by Lemma~\ref{InARunPalPair}). Since $b_2 - b_1 + 1 = d$, it follows from Lemmas~\ref{PrimitivePalPair} and~\ref{InARunPalPair} that there are at most two positions $h \in [b_1..b_2]$ such that $t[h{-}d..h{-}1] = v'u'$ for some palindromes $u'$ and $v'$ satisfying $|u'| = |u|$ and $|v'| = |v|$; these positions $h$ can be easily found using the equality from Lemma~\ref{InARunPalPair} and the lengths $|u''|$ and $|v''|$.

So, fix one such position $h\in [b_1..b_2]$. The position $h$ is a suspected starting position of $s_2$ in an instance of $p$. By the procedure similar to that used in the proof of Lemma~\ref{PalPattern}, we compute in $O(1 + \frac{d}{\log n})$ time a bit array $D'[h - |s_1| - d .. h - |s_1|]$ such that $D'[h'] = 1$ iff there is an instance $s_1\lvec{w}s_2ws_3$ of $p$ starting at position $h'$ and such that $h' + |s_1\lvec{w}| = h$. It follows from Lemma~\ref{AllInRunOccurrences} that any string $t[i..i + |s_1s_2s_3| + 2y]$ such that $y \ge d$, $i' \le i \le i + |s_1s_2s_3| + 2y \le j'$, and $i + |s_1| + y = h$ is an instance of $p$ iff the string $t[i + \lfloor\frac{y}d\rfloor d .. j - \lfloor\frac{y}d\rfloor d]$ is an instance of $p$, i.e., iff $D'[i + \lfloor\frac{y}d\rfloor d] = 1$. So, once the array $D'$ is computed, one can easily count the number of all instance $t[i..i + |s_1s_2s_3| + 2y]$ of $p$ such that $y \ge d$, $i' \le i \le i + |s_1s_2s_3| + 2y \le j'$, $i + |s_1| + y = h$. It remains to count the number of instance $t[i..i + |s_1s_2s_3| + 2y]$ of $p$ such that $y \ge d$, $i + |s_1| + y = h$, and either $i < i' \le i + |s_1|$ or $j - |s_r| \le j' < j$. We can do this with the same case analysis as for the case $p \ne s_1\lvec{x}s_2xs_3$.

\paragraph{In-two-runs instances of $p$.}

Like in the case of one run, our algorithm for two runs processes each run $t[i'..j']$ with period $d$ (only once) and finds all instances of $p$ whose substitutions have length ${\ge}3d$ and lie in exactly two runs: $t[i'..j']$ and another run with period $d$.

Choose $z \in [1..r)$. Let $t[i..j] = s_1w_1s_2w_2\cdots w_{r-1}s_r$ be an instance of $p$ such that $|w_1|{=}\cdots{=}|w_{r-1}| \ge 3d$, $t[i + |s_1| .. h - 1]$ is a substring of $t[i'..j']$, and $t[h + |s_{z+1}| .. j - |s_r|]$ is a substring of another run with period $d$, where $h = i + |s_1w_1\cdots s_zw_z|$. We call $z$ a \emph{separator} in $t[i..j]$. Obviously, the string $w_zs_{z+1}w_{z+1}$ does not have period $d$. Hence, by Lemma~\ref{wsw2}, we have $h \in (j'{+}1{-}d..j'{+}1]$ or $h \in (j'{-}|s_{z+1}|{-}d..j'{-}|s_{z+1}|]$ or $h = j' - |\pref_d(s_{z+1})| + 1$. Suppose that $h \in (j'{+}1{-}d..j'{+}1]$ (the other cases are similar). Let $b_1 = j' + 2 - d$ and~$b_2 = j' + 1$.

Since $|w_{z+1}| \ge 3d$, the string $t[h{+}|s_{z+1}|..h{+}|s_{z+1}w_{z+1}|{-}1] = w_{z+1}$ contains the substring $t[b_2 + |s_{z+1}| .. b_2 + |s_{z+1}| + 2d - 1]$. So, using Lemma~\ref{SubstringRun} for the latter substring, we find in $O(1)$ time a run $t[i''..j'']$ with period $d$ containing $w_{z+1}$. Clearly, the strings $t[i..h{-}1]$ and $t[h{+}|s_{z+1}|..j]$ are instances of the patterns $s_1x_1\cdots s_zx_z$ and $x_{z+1}s_{z+2}\cdots x_{r-1}s_r$, respectively, and $t[i{+}|s_1|..h{-}1]$ and $t[h{+}|s_{z+1}|..j{-}|s_r|]$ are substrings of the runs $t[i'..j']$ and $t[i''..j'']$, respectively. Hence, if either there is $z' \in (1..r) \setminus \{z{+}1\}$ such that $x_{z'-1} = x_{z'}$ or there are $z',z'' \in (1..r) \setminus \{z{+}1\}$ such that $x_{z'-1}s_{z'}x_{z'} = \lvec{x}s_{z'}x$ and $x_{z''-1}s_{z''}x_{z''} = xs_{z''}\lvec{x}$, then the number $|w_1| \bmod d$ is equal to one of the values described in Lemma~\ref{wmodd}; let $\delta' \in [0..d)$ be one of these values (we process each such $\delta'$). Otherwise (if we could not find such $z'$ and $z''$), we have $r \le 5$ and we can compute a similar value $\delta'$ as follows. If $p \ne s_1xs_2\lvec{x}s_3$ and $p \ne s_1\lvec{x}s_2xs_3$, then there are $z' \in (1..z]$ and $z'' \in (z{+}1..r)$ such that $x_{z'} = x_{z''}$. Denote by $\ell$ and $\ell'$ the starting positions of Lyndon roots of $t[i'..j']$ and $t[i''..j'']$, respectively; $\ell$ and $\ell'$ can be computed in $O(1)$ time by Lemma~\ref{LyndonRoots}. It follows from Lemma~\ref{PrimitiveVV} that $i + |s_1w_1\cdots s_{z'}w_{z'}| - \ell \equiv i + |s_1w_1\cdots s_{z''}w_{z''}| - \ell' \pmod{d}$. Therefore, $|s_{z'+1}w_{z'+1}\cdots s_{z''}w_{z''}| \equiv \ell' - \ell \pmod{d}$ and hence $(z'' - z')|w_1| \equiv \ell' - \ell - |s_{z'+1}s_{z'+2}\cdots s_{z''}| \pmod{d}$. This equation has at most $z'' - z'$ solutions: $|w_1| \equiv \frac{cd + \ell' - \ell - |s_{z'+1}s_{z'+2}\cdots s_{z''}|}{z'' - z'} \pmod{d}$ for $c \in [0..z''{-}z')$; let $\delta' \in [0..d)$ be one of these solutions (we process all such $\delta'$; since $z'' - z' \le r - 3 \le 2$, there are at most two such $\delta'$). Now let us consider the cases $p = s_1xs_2\lvec{x}s_3$ and $p = s_1\lvec{x}s_2xs_3$.

Suppose that $p = s_1\lvec{x}s_2xs_3$ (the case $p = s_1xs_2\lvec{x}s_3$ is symmetrical). Consider an instance $t[i..j] = s_1\lvec{w}s_2ws_3$ of $p$ whose substitutions of $x$ and $\lvec{x}$ have length ${\ge}3d$ and lie in distinct runs $t[i'..j']$ and $t[i''..j'']$ with period $d$. As above, by Lemma~\ref{wsw2}, we have $h \in (j'{+}1{-}d..j'{+}1]$ or $h \in (j'{-}|s_{z+1}|{-}d..j'{-}|s_{z+1}|]$ or $h = j' - |\pref_d(s_{z+1})| + 1$, where $h = i + |s_1\lvec{w}|$. Suppose that $h \in (j'{+}1{-}d..j'{+}1]$ (other cases are analogous) and denote $b_1 = j' + 2 - d$ and $b_2 = j' + 1$.

Denote by $\ell$ and $\ell'_0$ the starting position of a Lyndon root of $t[i'..j']$ and the ending position of a reversed Lyndon root of $t[i''..j'']$, respectively; $\ell$ and $\ell'_0$ can be found in $O(1)$ time by Lemma~\ref{LyndonRoots}. In order to synchronize parts of $p$ that are contained in $t[i'..j']$ and $t[i''..j'']$, we check in $O(1)$ time using the $\lcp$ structure whether $\lrange{t[\ell..\ell{+}d{-}1]} = t[\ell'_0{-}d{+}1..\ell'_0]$; if not, then there cannot be any instances of $p$ such as $t[i..j]$. It follows from Lemma~\ref{PrimitiveVV} that $h - \ell \equiv \ell'_0 - (h + |s_2| - 1) \pmod{d}$. Hence, we obtain $2h \equiv \ell + \ell'_0 - |s_2| + 1 \pmod{d}$. Since $b_2 - b_1 + 1 = d$, we can find in $O(1)$ time at most two positions $h$ in $[b_1..b_2]$ satisfying the latter equality. Fix one such position $h_0 \in [b_1..b_2]$; $h_0$ is a suspected starting position of $s_2$ in an instance of $p$.

Applying Lemma~\ref{PalPattern} with $h_1 = h_0 - d, h_2 = h_0 - 1, q = h_0 + |s_2|$ (see Fig.~\ref{fig:xslvecx}), we compute a bit array $occ[h_1 - d - |s_1| .. h_1 - |s_1|]$ such that, for any $h' \in [h_1 - d - |s_1| .. h_1 - |s_1|]$, we have $occ[h'] = 1$ iff $t[h'..n]$ has a prefix $s_1\lvec{w}s_2ws_3$ such that $h' + |s_1\lvec{w}| = h_0$. Since $\lrange{t[\ell..\ell{+}d{-}1]} = t[\ell'_0{-}d{+}1..\ell'_0]$, by the definition of $h_0$, it follows that any string $t[i..j]$ such that $i' \le i$, $j \le j''$, and $i + |s_1| + \ell_w = h_0$, where $\ell_w = (j - i + 1 - |s_1s_2s_3|) / 2$, is an instance of $p$ iff $occ[h_0 - d - (\ell_w \bmod d) - |s_1|] = 1$. So, in this way we found all instances of $p$ that correspond to $h_0$ and do not cross the boundaries $i'$ and $j''$.

Now suppose that $t[i..j] = s_1\lvec{w}s_2ws_3$ is an instance of $p$ such that $i < i' \le i + |s_1|$ and $i + |s_1\lvec{w}| = h_0$ (the case $j > j''$ is symmetrical). By Lemma~\ref{wsw}, we have either $i + |s_1| \in [i'..i'{+}d)$ or $i + |s_1| = i' + |\suff_d(s_1)|$. First, we check whether $t[i' + |\suff_d(s_1)| - |s_1| .. h_0 + |s_2| + (h_0 - i' - |\suff_d(s_1)|) + |s_3| - 1]$ is an instance of $p$ in $O(1)$ time using the $\lcp$ structure and the arrays $D_1, D_2, D_3$. Secondly, we find all instances $t[i..j] = s_1\lvec{w}s_2ws_3$ of $p$ satisfying $i + |s_1\lvec{w}| = h_0$ and $i + |s_1| \in [i'..i'{+}d)$ using Lemma~\ref{PalPattern} with $h_1 = i' + d, h_2 = i' + 2d - 1, q = h_0 + |s_2| + (h_0 - i' - 2d)$ (see Fig.~\ref{fig:xslvecx}).

Denote $\delta = 3d + \delta'$. It follows from Lemma~\ref{SepCond} that any separator $z \in Z$ (resp., $z \in Z'$, $z \in Z''$) satisfies~\eqref{eq:subz} for $Z_0 = Z$ (resp., $Z_0 = Z'$, $Z_0 = Z''$). So, we find $O(1)$ ``suspect'' separators, by Lemma~\ref{subz}. We apply the following lemma to each of them and essentially obtain two bit arrays encoding compactly the occurrences of $p$.

\begin{lemma}
Given $z \in [1..r)$, two runs $t[i'..j']$ and $t[i''..j'']$ with period $d$, a number $\delta \ge d$, and a segment $[b_1..b_2] \subset [i'..j'{+}1]$ of length $d$, we can compute in $O(r + \frac{rd}{\log n})$ time the numbers $d', d'', d''', h', h'_0, h'', h''_0, a', a'', a'''$ and the bit arrays $E[b_1..b_2]$, $F[b_1..b_2]$ such that:
\begin{enumerate}
\item for any $h \in [b_1..h']$ (resp., $h \in (h'..h'_0]$, $h \in (h'_0..b_2]$), we have $E[h] = 1$ iff the strings $t[h - |s_1s_2\cdots s_z| - z(\delta + cd) .. h + |s_{z+1}s_{z+2}\cdots s_r| + (r-1-z)(\delta + cd) - 1]$, for all $c \in [0..d']$ (resp., $c \in [0..d'']$, $c \in [0..d''']$), are instances of $p$ and $i' \le h - |s_1s_2\cdots s_z| - z(\delta + cd) \le h + |s_{z+1}s_{z+2}\cdots s_r| + (r-1-z)(\delta + cd) - 1 \le j''$;
\item for any $h \in [b_1..h'']$ (resp., $h \in (h''..h''_0]$, $h \in (h''_0..b_2]$), we have $F[h] = 1$ iff the string $t[h - |s_1s_2\cdots s_z| - z(\delta + ad) .. h + |s_{z+1}s_{z+2}\cdots s_r| + (r-1-z)(\delta + ad) - 1]$, where $a = a'$ (resp., $a = a''$, $a = a'''$), is an instance of $p$ and $i' \le h - |s_2s_3\cdots s_z| - z(\delta + ad) \le h + |s_{z+1}s_{z+2}\cdots s_{r-1}| + (r-1-z)(\delta + ad) - 1 \le j''$.
\end{enumerate}
In addition, we find at most two instances $t[i_0..j_0] = s_1w_1\cdots w_{r-1}s_r$ of $p$ such that $|w_1| = \cdots = |w_{r-1}| \ge 3d$, $|w_1| \equiv \delta \pmod{d}$, $i' \le i_0 + |s_1| \le j_0 - |s_r| \le j''$, $i_0 + |s_1w_1\cdots s_zw_z| \in [b_1..b_2]$, and it is guaranteed that any instance $t[i..j] = s_1w_1\cdots w_{r-1}s_r$ of $p$ such that $i' \le i + |s_1| \le j - |s_r| \le j''$, $|w_1| = \cdots = |w_{r-1}| \ge 3d$, $|w_1| \equiv \delta \pmod{d}$, and $i + |s_1w_1\cdots s_zw_z| \in [b_1..b_2]$ either is encoded in the arrays $E, F$ or is represented by one of the additional instances.\label{InTwoRunsLemma}
\end{lemma}
\begin{proof}
Denote $p_1 = s_1x_1\cdots s_zx_z$ and $p_2 = x_{z+1}s_{z+2}\cdots x_{r-1}s_r$. We apply Lemma~\ref{GeneralInRunLemma} putting $p := p_1$ to compute the numbers $d'_1,d''_1,h'_1,h''_1,a'_1,a''_1$, the bit arrays $E_1[b_1..b_2]$, $F_1[b_1..b_2]$, and, if needed, one additional instance $t[i^1_0..j^1_0]$ of $p_1$ that altogether represent all instances $t[i..j] = s_1w_1\cdots s_zw_z$ of $p_1$ such that $|w_1| = \cdots = |w_z| \ge 3d$, $|w_1| \equiv \delta \pmod{d}$, $i' \le i + |s_1| \le j \le j'$, and $j + 1 \in [b_1..b_2]$. Similarly, putting $p := p_2$, we apply a symmetrical version of Lemma~\ref{GeneralInRunLemma} to obtain numbers $d'_2,d''_2,h'_2,h''_2,a'_2,a''_2$, bit arrays $E_2[b_1{+}|s_{z+1}| .. b_2{+}|s_{z+1}|]$, $F_2[b_1{+}|s_{z+1}| .. b_2{+}|s_{z+1}|]$, and, probably, one additional instance $t[i^2_0..j^2_0]$ of $p_2$ that together represent all instances $t[i..j] = w_{z+1}s_{z+2}\cdots w_{r-1}s_r$ of $p_2$ such that $|w_{z+1}| = \cdots = |w_{r-1}| \ge 3d$, $|w_{z+1}| \equiv \delta \pmod{d}$, $i'' \le i \le j - |s_r| \le j''$, and $i \in [b_1{+}|s_{z+1}|..b_2{+}|s_{z+1}|]$. We combine these to get all required instances of $p$ as follows.

To combine instances of $p_1$ and $p_2$ encoded in the arrays $E_1[b_1..b_2]$ and $E_2[b_1{+}|s_{z+1}|..b_2{+}|s_{z+1}|]$, we perform in $O(\frac{d}{\log n})$ time the bitwise ``and'' of these arrays and the bit array $D_{z+1}[b_1..b_2]$ and thus obtain a bit array $E[b_1..b_2]$. Further, we find in $O(\frac{d}{\log n})$ time one arbitrary position $h \in [b_1..b_2]$ such that $E[h] = 1$. In order to ``synchronize'' instances of $p_1$ and $p_2$, we check in $O(1)$ time using the $\lcp$ structure whether $t[h{-}d..h{-}1] = t[h{+}|s_{z+1}|..h{+}|s_{z+1}|{+}d{-}1]$, if $x_z = x_{z+1}$, or $\lrange{t[h{-}d..h{-}1]} = t[h{+}|s_{z+1}| .. h{+}|s_{z+1}|{+}d{-}1]$, if $x_z \ne x_{z+1}$; if not, then we fill $E$ with zeros. One can show that $E$ satisfies the conditions in the statement of the lemma provided $h' = \min\{h'_1,h'_2\}$, $h'_0 = \max\{h'_1,h'_2\}$, $d' = \min\{d'_1,d'_2\}$, $d''' = \min\{d''_1, d''_2\}$, $d'' = \min\{d''_1, d'_2\}$ if $h'_1 \le h'_2$, and $d'' = \min\{d'_1, d''_2\}$ if $h'_1 > h'_2$.

We apply a similar analysis for all remaining combinations: $E_1$ and $F_2$, $F_1$ and $E_2$, $F_1$ and $F_2$; but due to the definitions of the arrays $F_1, F_2$ and the numbers $a'_1,a''_1,a'_2,a''_2,h''_1,h''_2$, we can combine the results into one bit array $F[b_1..b_2]$ putting $h'' = \min\{h''_1,h''_2\}, h''_0 = \max\{h''_1, h''_2\}, a' = \min\{a'_1, a'_2\},a''' = \min\{a''_1, a''_2\}, a'' = \min\{a''_1,a'_2\}$ if $h''_1 \le h''_2$, and $a'' = \min\{a'_1, a''_2\}$ if $h''_1 > h''_2$. Finally, we try to ``extend'' in an obvious way the instance $t[i^1_0..j^1_0]$ of $p_1$ (similarly, $t[i^2_0..j^2_0]$ of $p_2$) to a full instance of $p$ in $O(r)$ time using the $\lcp$ structure and the arrays $\{D_z\}_{z=1}^r$. Thus, we obtain at most two instances of $p$.\qed
\end{proof}

\end{document}